\documentclass{ieeeaccess}
\usepackage{cite}
\usepackage{amsmath,amssymb,amsfonts}
\usepackage{algorithmic}
\usepackage{graphicx}
\usepackage{textcomp}

\usepackage{verbatim}
\usepackage{kotex}
\usepackage[utf8]{inputenc}
\usepackage{subfigure}
\usepackage{dsfont}
\usepackage{amssymb}
\usepackage{algorithm}
\usepackage{cuted}
\usepackage{epstopdf}
\usepackage{soul,color}
\usepackage{setspace}
\usepackage{multirow}
\usepackage{setspace}
\usepackage{lipsum}
\usepackage{multicol}
\usepackage{float}

\usepackage{makecell}

\usepackage[top=0.61in, bottom=0.61in, left=0.60in, right=0.60in]{geometry}
\usepackage{amsthm}

\newtheorem{pro}{Proposition}

\newcommand\blfootnote[1]{%
  \begingroup
  \renewcommand\thefootnote{}\footnote{#1}%
  \addtocounter{footnote}{-1}%
  \endgroup
}
\DeclareMathOperator*{\argmin}{argmin}

\def\BibTeX{{\rm B\kern-.05em{\sc i\kern-.025em b}\kern-.08em
    T\kern-.1667em\lower.7ex\hbox{E}\kern-.125emX}}
\begin{document}
\history{Received July 17, 2024, Accepted July 24, 2024, date of current version July 24, 2024.\\
This work has been submitted to the IEEE for possible publication. Copyright may be transferred without notice, after which this version may no longer be accessible.}
\doi{10.1109/ACCESS.2024.3434665}

\title{
Guided-Mutation Genetic Algorithm for Mobile IoT Network Relay 
}
\author{
	\uppercase{Gyupil Kam}\authorrefmark{1} \IEEEmembership{Member, IEEE} 
	and \uppercase{Kiseop Chung}\authorrefmark{2} \IEEEmembership{Member, IEEE}
}
\address{Agency for Defense Development, Daejeon 34186, Republic of Korea (e-mail: halogen363@gmail.com\authorrefmark{1}, cks991030@snu.ac.kr\authorrefmark{2})}
\tfootnote{This work was supported by the Agency for Defense Development Grant funded by the Korean Government (2024).}

\markboth
{G.P. Kam \headeretal: Guided-Mutation Genetic Algorithm for Mobile IoT Network Relay}
{G.P. Kam \headeretal: Guided-Mutation Genetic Algorithm for Mobile IoT Network Relay}

\corresp{Corresponding author: Kiseop Chung}

\begin{abstract}

The Internet of Things (IoT) is a communication scheme which allows various objects to exchange several types of information, enabling functions such as home automation, production management, healthcare, etc.
In addition, energy-harvesting (EH) technology is considered for IoT environment in order to reduce the need for management and enhance maintainability. 
Moreover, since environments considering outdoor elements such as pedestrians, vehicles and drones have been on the rise recently, it is important to consider mobility when designing an IoT network management scheme.
However, calculating the optimal relaying topology is considered as an NP-hard problem, and finishing computation for mobility environment before the channel status changes is important to prevent delayed calculation results. 
In this article, our objective is to calculate a sub-optimal relaying topology for stationary and mobile system within reasonable computation time. 
To achieve our objective, we validate an iterative balancing time slot allocation algorithm introduced in the previous study, and propose a guided-mutation genetic algorithm (GMGA) that modulates the mutation rate based on the channel status for rational exploration. 
Additionally, we propose a mobility-aware iterative relaying topology algorithm, which calculates relaying topology in a mobility environment using an inheritance of the sub-optimal relaying topology calculations.
Simulation results verify that our proposed scheme effectively solves formulated IoT network problems compared to other conventional schemes, and also effectively handles IoT environments including mobility in terms of minimum rate budget and computation time.

\end{abstract}

\titlepgskip=-15pt

\maketitle

%%%%%%%%%%%%%%%%%%%%%%%%%%%%%%%%%%%%%%%%%%%%%%%%%%%%%%%%%%%%%%%%%%%%%%%%%%%%%
\section{Introduction}
\label{sec:introduction}

\blfootnote{
The preprint of this article may be found in:\\ 
http://arxiv.org/abs/2404.01683
}

The Internet of Things (IoT) has opened up a world of possibilities in several areas such as smart grid, home automation, smart transportation, etc \cite{Khanna20}. 
With an ability for distant things to exchange duplex information in a wide area, the limitations of distance and time for various fields have been reduced.
In addition, IoT facilitates convenient access and interaction with a diverse range of devices, including home appliances, surveillance cameras, actuators, sensors, and vehicles. 
This capability encourage the creation of numerous applications that leverage the vast and varied data produced by these devices to offer new services to individuals, businesses, and government entities \cite{Kirimtat20}.

One of the applications of IoT is the smart cities, which can use public resources more efficiently and improve the quality of services for citizens while reducing operational costs through device connectivity \cite{Zanella14}. 
The concept of smart cities solves the challenges related to transportation, environment, and social cohesion occurring from the increase in urban population and complexity of city.
Thus, smart cities are expected to become a promising business in the near future \cite{Yang21}.

While maintaining such a system, recharging the batteries of large-scale battery-based wireless sensors requires substantial manpower. 
Also, battery has limitations by battery power or capacity to support lifetime.
Moreover, battery cannot be replaced in an active network or during the data processing phase \cite{Ahmad22}. 
In order to handle the problem related to batteries, an energy harvesting (EH) technology provides an effective power management methodology in terms of usability and resilience, using power beacons (PBs) to supply energy to network nodes by transmitting solar or RF energy \cite{Sanislav21}.
However, since the energy distribution and the position of nodes may be uneven under the variety of distances between the nodes and the PBs, finding the optimal relaying topology under these circumstances is complicated.

Moreover, since mobile nodes such as vehicles continuously alter positions over time, recalculations of relaying topology according to these changes are necessary.
Otherwise, the relaying topology does not properly reflect the channel characteristics of the network, eventually resulting in a deterioration in the overall performance of the network \cite{Kim20}. 
Given these challenges, finding the optimal topology for this kind of problem is NP-hard and requires extensive computation time \cite{Zhang20}.

However, despite the NP-hardness of the task, determining the relaying topology is essential for managing the rate budget of system because relaying can significantly impact the overall rate budget, either positively or negatively.
In addition, time slot allocation is also important for time division multiple access (TDMA) system, a technology that divides the frequency bands available in a network into time slots to allow multiple users to share the same transmission medium while avoiding interference \cite{Benrebbouh21}. 
Therefore, optimizing the relaying topology and allocating time slots effectively is crucial for maximizing the rate budget of the system.

Besides, several methods such as genetic algorithms have been applied in order to solve various meta-heuristic problems. 
A genetic algorithm which consists of crossover, mutation, evaluation, and selection process, is an appropriate tool to solve NP-hard problems. 
Through different modulations, genetic algorithm offers a wide range of benefits and can address various types of meta-heuristic challenges \cite{Katoch21}. 
Therefore, by adapting the genetic algorithm specifically to our problem, a sub-optimal rate budget within a minimal computation time could be achieved.

In this paper, we propose a guided-mutation genetic algorithm (GMGA) including modulation of mutation process, and mobility-aware iterative relaying topology algorithm to derive a sub-optimal relaying topology by solving the formulated problem under IoT network and EH environment. 
To establish the theoretical foundation of previous research, we prove the validity of the iterative balancing time slot allocation algorithm, which was proposed in previous research in order to maximize the minimum rate budget of a node in the system by adjusting the time slot allocation.
Numerical analysis confirms that not only our proposed methodology outperforms than conventional applicable methods for our formulated problem, but also our proposed methodology is applicable for a real-world implementation considering mobility.

\subsection{Contributions and Organization}

The main contributions of this paper are threefold:

\begin{itemize}
    \item First, we prove that for the Iterative Balancing time slot allocation algorithm which allocates time slots in order to maximize minimum rate budget (bits/Hz), equivalence of maximum and minimum bits/Hz equals to the maximization of minimum bits/Hz using dual problem and KKT conditions.  
    \item Second, we propose a guided-mutation genetic algorithm (GMGA), which effectively handles both scalability and non-linearity of the formulated IoT environment based on EH.
    By mutating the node connections with a probability based on the cost of each links, our GMGA shows superior performance results than conventional algorithm in less computation time. 
    \item Third, we propose a mobility-aware iterative relaying topology algorithm to handle environments considering mobility, which reduces repetitive calculations and derives sub-optimal topology that alters with movement of the nodes, thus enhancing adaptability into the real-world environments such as outdoor transportation.
\end{itemize}

The remainder of this paper consists of the following.
In section II, we review the previous research related to main theme of this paper, which consists of previously proposed network relaying schemes and its methodology, along with GA and its modulation.
In Section III, we establish the system model and formulate an IoT optimization problem, which is a main problem in this paper. 
In Section IV, we prove the validity of IB time slot allocation algorithm and propose a GMGA in order to calculate relaying topology.
In addition, we propose a mobility-aware iterative relaying topology algorithm based on the formulated system model. 
In Section V, we provide a detailed explanation of simulation parameters and conduct simulations to confirm the sub-optimality of our proposed scheme by comparing performance with other conventional schemes. 
Finally, in Section VI, we provide the key findings of our research and conclude the paper with a short summary. 

\section{Related Works}

Recent studies have focused on developing relaying schemes to calculate the optimal relaying topology.
In \cite{Hung20}, an energy-efficient relaying scheme for wireless sensor networks (WSNs) was explored, emphasizing relaying methods in scenarios where event detection occurs among numerous nodes. 
Additionally, \cite{Ghazi17} devised a `Teaching Learning' based optimization algorithm to navigate the optimal path through a selection of nodes for relaying. 
\cite{Xu23} presents a novel relay selection strategy, named the WF strategy, utilizing Markov models and queuing theory to enhance transmission efficiency and performance in WSNs by selecting optimal relay nodes based on multimedia service characteristics. 
The simulation results demonstrate that the WF strategy outperforms existing weighted round robin and random selection strategies, achieving better system performance metrics such as throughput and packet loss.

Furthermore, there is also some studies applying machine learning on relaying schemes.
\cite{Vu22} presents in-depth analysis of the performance of Non-Orthogonal Multiple Access (NOMA) in IoT applications, particularly focusing on wireless powered cognitive relay networks. 
The analysis is dedicated to various deep learning techniques that predict and optimize the performance of the networks. 
\cite{Mao17} attempted to find a relaying topology using a supervised learning, which is actually inefficient in complex scenarios. 
Since calculating relaying topology is an NP-hard problem, application of supervised learning is often limited due to the requirement for extensive labeled training data.
To overcome this issue, \cite{Chung23} calculated a sub-optimal relaying topology using variational autoencoder (VAE) with novel evaluation method which does not require a labeled training dataset, and proposed iterative balancing time slot allocation algorithm.
The VAE scheme proposed in \cite{Chung23} applies both training phase and inference phase for each node distribution, making it inseparable these two phases. 
In other words, instead of using the inference phase of a pre-trained model, the VAE scheme uses both the training phase and inference phase to obtain the relaying topology of a single node distribution.
Thus, the VAE scheme reached its limit due to time-consuming nature of backward propagation included in the training phase.

There are also other studies to solve the problem through reinforcement learning \cite{Xu21}. 
However, the lack of proven convergence and the extensive time required to discover the sub-optimal path remain significant challenges \cite{Mammeri19}. 
In addition, since IoT network has the characteristic of low transmission frequency, using reinforcement learning for such IoT network results in very low frequency of reward, which hinders applying reinforcement learning to such network. 
Moreover, the reinforcement learning approach necessitates using the imperfect topology in the physical network during the initial learning phase, which inevitably leads to a practical performance degradation.

In addition to research on stationary nodes, research on mobile nodes has also been conducted in WSNs. 
In \cite{Piltyay20}, the connectivity within heterogeneous 5G mobile systems is investigated with an analysis of how connectivity changes based on the characteristics of mobile nodes and network parameters. 
In order to investigate connectivity, they proposed three algorithms for the cluster analysis of WSNs:  k-means, algorithm of the foreign element, fuzzy C-means algorithm. 
\cite{Haseeb19} studied an intrusion prevention framework to improve secure routing in mobile IoT environments using WSN, focusing on network lifetime, data reliability, and security. 
\cite{Soltani20} proposed a random waypoint model incorporating random orientation of user equipment, applied to analyze handover rates in indoor LiFi networks.
\cite{Sadrishojaei21} introduces a novel routing method for the mobile IoT, focusing on energy-efficient clustering and predictive location routing with multiple mobile sinks, which significantly improves energy consumption and network throughput. 
Through simulations, the method demonstrates at least a 28.12\% reduction in energy usage and a 26.74\% increase in throughput compared to existing energy-efficient routing algorithms.
These studies on mobile nodes have not proceeded with the relaying topology due to the considerable amount of required computation resources. 
In summary, despite the development of various schemes, persistent limitations have necessitated ongoing research in this area.

In connectivity-related issues, genetic algorithm (GA) is frequently employed across various applications within the telecommunications division.
\cite{Tao18} researched network security for intrusion detection using GA and support vector machines (SVMs), demonstrating a higher detection rate compared to other SVM-based intrusion detection algorithms.
\cite{Zhang19} proposed a novel method combining an improved genetic algorithm with a deep belief network (DBN) to enhance intrusion detection in IoT environments. 
GA has also been applied for ad-hoc optimization, as seen in \cite{Pan21}, where GA was used for UAV path planning to identify the optimal UAV relay, and as seen in \cite{Hanh19}, which applied GA to solve the area coverage problem in WSN. 
\cite{Khadir22} represents a case where GA was utilized to efficiently compute dynamic QoS parameters in IoT systems concerning computation time and optimality. 
Furthermore, \cite{Bouzid20} developed a multi-objective wireless network optimization using the genetic algorithm (MOONGA) to propose a method for identifying the optimal node positions in fixed relaying topology considering multiple objectives such as sensing coverage, connectivity, lifetime, energy consumption, and cost. 
However, node locations are often constrained by practical needs or geographical limitations in reality, necessitating algorithms that calculate the relaying topology while holding the node positions fixed.
\cite{Birtane24} discusses the application of the Vibrational Genetic Algorithm (VGA) for optimizing the deployment of heterogeneous sensor nodes in irregularly shaped areas of wireless sensor networks. 
VGA enhances the placement strategy by applying small, random vibrations to the positions of nodes during the optimization process, which helps in avoiding local optima and improves coverage in complex terrains. 
The algorithm's effectiveness is demonstrated through simulations showing improved coverage and connectivity compared to traditional methods.

Since the method of applying GA to a problem varies depending on the problem, special modulations are applied to GA. 
\cite{Song18} applied ranking group selection, direction-based crossover, and normal mutation to enhance the search capabilities of algorithm and population diversity, thereby improving solutions to complex optimization problems. 
\cite{Wang23} increased search efficiency for continuous constrained optimization problems through two-direction crossover and grouped mutation. 
\cite{Vlasic19} experimentally demonstrated that scheduling problems could be solved using dispatching rules for population initialization.
NSGA-II (Nondominated Sorting Genetic Algorithm II) used in \cite{Guo23} is employed to optimize the efficiency and volume of magnetic components in DAB converters by iteratively applying genetic operations and nondominated sorting to evolve a population towards a set of Pareto optimal solutions, which balance trade-offs between efficiency and volume. 
However, these modulations are limited to the problems proposed in each paper, and to solve our connectivity problem, a modulated GA with reduced computation time and for discrete values is necessary. 

Based on the mentioned works summarized in Table \ref{schemes}, a methodology that calculates relaying topology of stationary nodes or mobile nodes with a sufficient rate budget and a fast computation time has yet to be proposed. 
However, various studies on GA have shown the potential for applying GA in IoT environments, and have achieved high performance through appropriate modulation. 
Therefore, we aim to calculate the relaying topology based on environments with stationary or moving nodes using a genetic algorithm with mutation rate modulation, and evaluate the performance in terms of rate budget and computation time in this paper.

\begin{table*}[]
\centering
\setlength{\tabcolsep}{10pt}
\renewcommand{\arraystretch}{1.3}
\small
\begin{tabular}{cc}
\Xhline{3\arrayrulewidth}
Algorithms & Explanation \\
\Xhline{3\arrayrulewidth}
Supervised deep belief architectures {\cite{Mao17}} & Using ANN based on supervised learning method \\
VAE scheme {\cite{Chung23}} & Using variational autoencoder with PT-EVM \\
Multi-agent DRL methods {\cite{Xu21}} & Using reinforcement learning with multi agent \\
MOONGA {\cite{Bouzid20}} & Using GA for the optimization of location of nodes \\
VGA {\cite{Birtane24}} & Using GA with vibrations to the positions of nodes \\
\Xhline{3\arrayrulewidth}
\end{tabular}
\caption{Summary of relaying topology schemes in recent works.}
\label{schemes}
\end{table*}

%%%%%%%%%%%%%%%%%%%%%%%%%%%%%%%%%%%%%%%%%%%%%%%%%%%%%%%%%%%%%%%%%%%%%%%%%%%%%
\section{System Model and Problem Formulation}

\begin{figure}[t]
	\centering
	\includegraphics[width=0.90\linewidth]{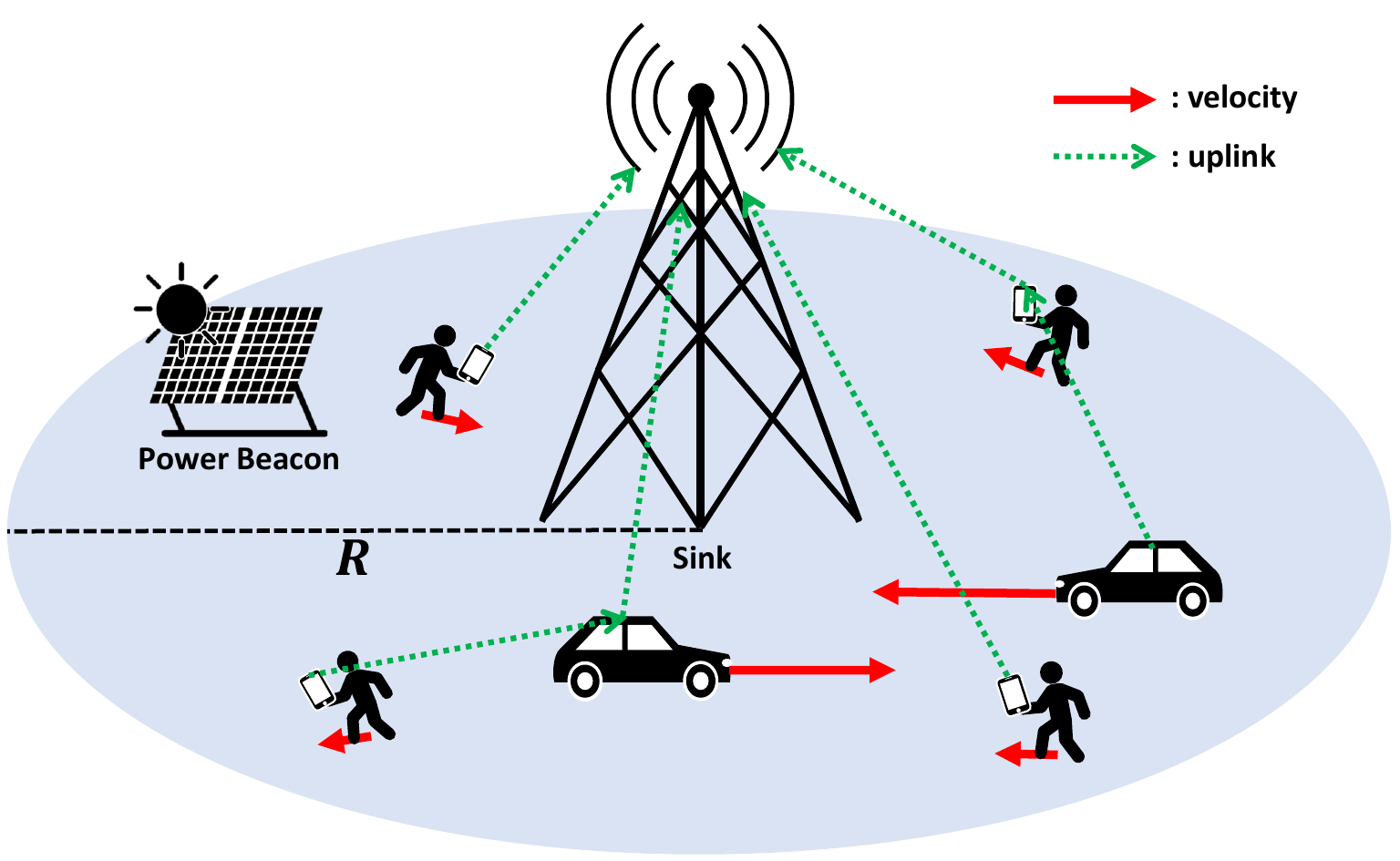} 
	\caption{System model example when $N_d=6$, $N_b=1$.} 
	\label{system_model}
 \vspace{-0.5 cm}
\end{figure}

\subsection{System Model}

This subsection describes a system model.
The basic configuration of the system model is inspired by approach of \cite{Chung23}.
In our system model, we consider $N_d$ IoT devices (nodes) and $N_b$ power beacons (PBs) which are uniformly distributed within a circle with a radius R, centered around a single data packet destination (packet sink). 
By labeling the nodes as 1,2,...,$N_d$ and the sink as $N_{d}+1$, the set associated with the nodes is $\mathbb{N}_d = \{1,2,...,N_d\}$ and the set associated with the overall communicating nodes is $\mathbb{N} = \{1,2,...,N_{d}+1\} = \mathbb{N}_{d} \cup \{N_{d}+1\}$, respectively.

Furthermore, since TDMA system has abilities of high energy efficiency and an on-demand feature, we consider time division multiple access (TDMA) system in our system model.  
Those abilities are necessary for IoT devices operating with energy harvesting (EH) and offer flexibility in resource allocation for various devices \cite{Piyare18}. 
In our TDMA-based system model, each node transmits data to the higher node defined by relaying topology in assigned time slot within time frame $T$ on the basis of TDMA system.

The time frame $T$ in the TDMA system is divided into $N_d$ time slots, and the time slot assigned to a node $i$ is denoted as $t_i$ ($i \in \mathbb{N}_d$).
The time slot $t_i$, allocated for node $i$ to send the data to the upper node, is defined by the following conditions:
\begin{align}
\sum_{i\in\mathbb{N}_d}{t_i}=T,\,\, t_i\in (0,T)\,\, \forall i\in\mathbb{N}_d.
\end{align}

Additionally, EH based on radio frequency (RF) is considered to supply energy to the nodes in our system model. 
For an RF power source, the harvested energy delivered to a \textit{$k$}th node is $E_{eh,k}=T\sum_{n\in\mathbb{N}_b}{P_b |h_{n,k}|^2 d^{-\alpha}_{n,k}}$, where $P_b$ is power transmitted by PB, $h_{n,k}$ is a channel between \textit{$n$}th PB and \textit{$k$}th node following $\mathcal{CN}(0,1)$, $d_{n,k}$ is distance between \textit{$n$}th PB and \textit{$k$}th node, and $\alpha$ is a path loss exponent, respectively. 
Note that the network model is based on a first-order radio model without consideration of circuit power consumption, as it will only result in a slight difference in the degree of improvement in our proposal.
Also, clock synchronization might be implemented by allocating fixed part of time slot and utilizing it to exchange the information about clock synchronization. 
However, we assume clock is synchronized in this paper for simplicity of simulation. 

Assuming that all the energy delivered to the \textit{$k$}th node is consumed for transmission of data within its time slot, the amount of rate budget (bits/Hz), $r_{k}(\mathbf{c},\mathbf{t})$, of \textit{$k$}th node transmitting to a parent node under TDMA system is given by the following according to Shannon's capacity:

\begin{align}
    r_{k}(\mathbf{c},\mathbf{t})&=\sum_{n\in\mathbb{N}}{c_{k,n}t_{k}} \log_{2}{(1+\frac{E_{k}|h_{k,n}|^2 d_{k,n}^{-\alpha}}{t_k N})}, \\
    c_{k,n}&=
    \begin{aligned}
        \begin{cases}
            1, \quad \mbox{if}\,\, n\text{th node is parent of } k\text{th node}, \\
            0, \quad \text{otherwise}.
        \end{cases}
    \end{aligned}
\end{align}
where $\mathbf{t}$ is an allocated time slot vector, $E_k$ is energy of \textit{$k$}th node by EH which is proportional to $E_{eh,k}$, and $N$ is noise power, respectively. 
However, since it is necessary to transmit not only own data but also data from child nodes, the amount of bits/Hz budget for transmitting data of \textit{$k$}th node denoted as $R_k(\mathbf{c},\mathbf{t})$ is as follows:
\begin{align}
    R_k(\mathbf{c},\mathbf{t})=\sum_{n\in\mathbb{N}}{c_{k,n}t_{k}} \log_{2}{(1+\frac{A_{k,n}}{t_k})} \nonumber\\
    -\sum_{n'\in\mathbb{N}_d}{c_{n',k}t_{n'}} \log_{2}{(1+\frac{A_{n',k}}{t_{n'}})}, \label{f_calculate}
\end{align}
where $A_{k,n}=E_{k}|h_{k,n}|^2 d_{k,n}^{-\alpha}/{N}$ and the signal to noise ratio (SNR) is $\Gamma_{k,n}=A_{k,n}/t_k$ when transmitting node is \textit{$k$}th node and receiving node is \textit{$n$}th node. 

For a simulation of mobile nodes, we assume that nodes are moving within the circle according to the Random Waypoint Model (RWM) which was adopted in mobile Ad-hoc network (MANET) simulation of \cite{Naik19}. 
Note that in our simulation, PBs are stationary; however, applying mobility to them is not difficult and does not critically impact the simulation results.
Since mobility models do not consider whether the waypoint of the next step is inside the boundary of the circle, movements derived from some mobility models could violate the boundary condition.
However, RWM is free from this consideration due to the convexity of the circle, resulting the derived waypoint of the inbound point from RWM is also inbound.

Moreover, we define a unit mobility simulation time $T_u$ which corresponds to the update period of node positions. 
In our system model, $T_u$ is defined as an integer multiple of $T$ ($T_u=kT$, $k \in \mathbb{Z}^{+}$).
Furthermore, all nodes move at a constant speed $v_c$ within the given circle with a radius $R$. 
In the real-world, nodes might have varying speeds over time, but such consideration on our system model causes only minor impact in this paper.

\subsection{Problem Formulation}
In this subsection, we formulate our target problem based on the system model defined in the subsection above. 
Our objective is to maximize $R_k(\mathbf{c},\mathbf{t})$ while each node transmits data to the sink within a limited time frame, $T$ in the system model. 
In order to increase the bits/Hz in a systematic aspect rather than an individual aspect, it is important to enhance the capacities for all nodes since the minimum $R_k(\mathbf{c},\mathbf{t})$ limits overall performance of the system. 
Therefore, the main objective to enhance the performance of the system model is to maximize the minimum value among all $R_k(\mathbf{c},\mathbf{t})$ of nodes.

However, since $R_k(\mathbf{c},\mathbf{t})$ takes into account both time slot $\mathbf{t}$ and connectivity $\mathbf{c}$, we include constraints on $\mathbf{t}$ and $\mathbf{c}$ into the main problem.
In our problem, $\mathbf{t}$ corresponds to the allocated time slots within the time frame T, where the sum of allocated time slots for all nodes equals to T.
Additionally, in our system model, every node has a single uplink route connected to the sink.
Considering all these constraints, our final formulated problem is as follows: 

\begin{align}
   \text{(P1)}  \quad \max_{\mathbf{c},\mathbf{t}} \, &\min_{k} \quad R_{k}(\mathbf{c},\mathbf{t})\label{p1}\\
\textrm{s.t.} \quad & \sum_{n \in \mathbb{N}_d} t_n =T, \label{c1}\\
  & \sum_{n'\in \mathbb{N}} C_{n,n'} =1\,\, \forall n \in \mathbb{N}_d,\label{c2}\\
  & (\mathbf{C}^{N_d})_{n,N_d+1} = 1\,\, \forall n \in \mathbb{N}_d \label{c3} .
\end{align}

where $\mathbf{C}$ is extended version of $\mathbf{c}$ satisfying 
\begin{align}
    &\mathbf{C}_{n,n'}=\mathbf{c}_{n,n'}\quad \forall n \in \mathbb{N}_d,\, n' \in \mathbb{N},\label{C1} \\
    &\mathbf{C}_{N_d+1,n'} = 0 \quad \forall n' \in \mathbb{N}_d, \label{C2}\\
    &\mathbf{C}_{N_d+1,N_d+1} = 1. \label{C3}
\end{align}

Then, (\ref{C1})-(\ref{C3}) provides a mathematical approach to connection, while (\ref{c2}) ensures that each node has a single uplink to a parent node, and (\ref{c3}) guarantees that there are no cycles in the system, and every node ultimately transmits data to the sink.

However, finding ($\mathbf{c}, \mathbf{t}$) which maximizes minimum of $R_{k}(\mathbf{c},\mathbf{t})$ is challenging due to following reasons: First, the presence of logarithmic operations in Shannon's capacity. Second, the NP-hardness of finding the optimal $\mathbf{c}$ since the number of possible candidate of $\mathbf{c}$ is $(N_d+1)^{N_d-1}$ in our formulated problem.

In this paper, we address each challenge from three aspects. 
First, an Iterative Balancing (IB) time slot allocation algorithm is used in order to calculate optimal $\mathbf{t}$ in given $\mathbf{c}$ \cite{Chung23}.
We also prove a validity of IB time slot allocation algorithm to increase the reliability of our formulated problem. 
Second, we propose a GMGA method to calculate sub-optimal solution of $\mathbf{c}$ in reasonable computation time.
Finally, we propose a mobility-aware iterative relaying topology algorithm by applying mobility in our system model in order to consider the real-world application.

%%%%%%%%%%%%%%%%%%%%%%%%%%%%%%%%%%%%%%%%%%%%%%%%%%%%%%%%%%%%%%%%%%%%%%%%%%%%%

\section{Proposed Methodology}

In this section, we propose a GMGA in order to calculate sub-optimal solution of $\mathbf{c}$, which modulates the mutation rate of GA based on a cost of each links.
We also propose a mobility-aware iterative relaying topology algorithm in order to handle wireless channel status changing over time without repetitive calculations.
It uses the final result from a prior simulation frame in the current simulation frame.
Moreover, in IB time slot allocation algorithm, we prove that equivalence of maximum $R_k(\mathbf{t})$ and minimum $R_k(\mathbf{t})$ achieves the maximization of minimum $R_k(\mathbf{t})$ with given $\mathbf{c}$ in order to strengthen the theoretical foundation.

\subsection{Proof of Validity for Iterative Balancing Time Slot Allocation Algorithm}
In this subsection, we describe IB time slot allocation algorithm assigning time slot to each node from the restricted time source $T$, regarding TDMA system.
When $\mathbf{c}$ is given, the IB time slot allocation algorithm works by allocating time $\Delta$ which is subtracted from a node with maximum $R_k(\mathbf{t})$, to time slot of node whose $R_k(\mathbf{t})$ is minimum, while maintaining the sum of time slots same. 
After finishing one iteration of time slot allocation, $\Delta$ is substituted for $\Delta/2$.
The convergence of the algorithm which is the difference between maximum $R_k(\mathbf{t})$ and minimum $R_k(\mathbf{t})$ converges to value smaller than $\epsilon_1$ is proven in \cite{Chung23}.

However, although decrease of difference between maximum and minimum ensures the equilibrium of $R_k(\mathbf{t})$, it does not guarantee maximization of minimum $R_k(\mathbf{t})$. 
In this paper, we prove that maximization of minimum $R_k(\mathbf{t})$ with constraints (\ref{c1}), (\ref{c2}), (\ref{c3}) is achieved when $R_1(\mathbf{t})=R_2(\mathbf{t})=...=R_{N_d}(\mathbf{t})$.
The proof of this equivalence is given as follows.

\begin{pro}
Let us define the amount of \rm bits/Hz \it the \textit{$k$}th node itself can transmit as $R_k(\mathbf{t})$ with given $\mathbf{c}$ and $k^*=\argmin\limits_{k}{R_k(\mathbf{t})}$. 
Then, a condition to maximize $R_{k^*}(\mathbf{t})$ where $t_1+t_2+...+t_{N_d}=T$ is $R_1(\mathbf{t})=R_2(\mathbf{t})=...=R_{N_d}(\mathbf{t})$. 
\end{pro}

\begin{proof}

First, we consider the characteristic of $R_k(\mathbf{t})$ function. 
Partial derivative of $R_k(\mathbf{t})$ with respect to $t_k$, denoted by $\nabla_{t_k} R_k(\mathbf{t})$, is calculated as follows:

\begin{align} \label{Rprime}
    \nabla_{t_k} R_k(\mathbf{t})&=\frac{\partial{R_k(\mathbf{t})}}{\partial{t_k}} \nonumber\\ 
    &=\sum_{n\in\mathbb{N}}{\frac{c_{k,n}}{\ln{2}}}{(-\frac{A_{k,n}}{t_k+A_{k,n}}+\ln{(1+\frac{A_{k,n}}{t_k})})}.
\end{align}
Since $\ln{(1+x)}>x/(1+x)$ when $x>0$, $\nabla_{t_k} R_k(\mathbf{t})>0$ holds, which means that $R_k(\mathbf{t})$ is monotonic increasing function for $t_k$.
This result implies that as more time is allocated to \textit{$k$}th node, the amount of data that the \textit{$k$}th node transmit to the upper node increases. 

In order to make our problem into a minimizing problem, let us define $S_k(\mathbf{t}) = -R_k(\mathbf{t})$.
Since (\ref{Rprime}) guarantees that $R_k(\mathbf{t})$ is monotonic increasing function for $t_k$, $S_k(\mathbf{t})$ is monotonic decreasing function for $t_k$:
\begin{align} \label{Sprime}
    \nabla_{t_k}S_k(\mathbf{t})<0,
\end{align}
and our formulated problem is as follows:
\begin{equation}
    \begin{aligned}
        \text{(P2)}\quad \argmin_{\mathbf{t}}\, \max_{k}\quad &S_k(\mathbf{t}) \\ 
        \textrm{s.t.} \quad &\sum_{j \in \mathbb{N}_d} {t_j}=T.
    \end{aligned}
\end{equation}
By adding inequality terms about the maximization of $S_k(\mathbf{t})$, we can remove an inner maxima:
\begin{equation}
    \begin{aligned}
        \text{(P3)}\quad \argmin_{\mathbf{t}}\quad &S_{k^*}(\mathbf{t}) \\
        \textrm{s.t.} \quad &S_{k^*}(\mathbf{t})\geq S_k(\mathbf{t})\,\, \forall\, k\in\mathbb{N}_d\backslash\{k^*\}, \\
        &\sum_{j \in \mathbb{N}_d} {t_j}=T.
    \end{aligned}
\end{equation}
Note that we exclude a condition when $k=k^*$ at the inequality constraint in order to simplify the calculation because excluding it does not affect the overall problem. 
Since $S_k(\mathbf{t})$ is not convex for $\mathbf{t}$, convex optimization method is not applicable.
Thus, we convert the problem into a dual problem using Lagrange function and find the optimal solution of primal problem using Karush–Kuhn–Tucker (KKT) conditions. 

Lagrange function of the problem (P3) is as follows:
\begin{align}
    L(\mathbf{t},\boldsymbol{\lambda},\nu)&=S_{k^*}(\mathbf{t})+\sum_{i \in \mathbb{N}_d \backslash \{k^*\}}{\lambda_i(S_i(\mathbf{t})-S_{k^*}(\mathbf{t}))}\nonumber\\
    &+\nu (\sum_{j \in \mathbb{N}_d} {t_j}-T),
\end{align}

and dual problem is as follows:
\begin{equation} \label{dualproblem}
    \begin{aligned}
    \max_{\boldsymbol{\lambda},\nu}\,\,\min_{\mathbf{t}}\quad&L(\mathbf{t},\boldsymbol{\lambda},\nu)\\
    \textrm{s.t.} \quad &\boldsymbol{\lambda} \succcurlyeq \mathbf{0}.
    \end{aligned}
\end{equation}

To find the optimal solution of primal problem (P3), KKT conditions are considered:
\begin{align}
    &\text{(P4-1)}\quad \nabla_\mathbf{t}L(\mathbf{t}^*,\boldsymbol{\lambda}^*,\nu^*)=\mathbf{0}, \\
    &\text{(P4-2)}\quad \lambda^*_i (S_i(\mathbf{t^*})-S_{k^*}(\mathbf{t^*}))=0 \,\,\,\forall\, i \in \mathbb{N}_d\backslash\{k^*\},\\
    &\text{(P4-3)}\quad S_i(\mathbf{t^*})-S_{k^*}(\mathbf{t^*}) \leq 0 \,\,\,\forall\, i \in \mathbb{N}_d\backslash\{k^*\},\\ 
    &\text{(P4-4)}\quad \sum_{j \in \mathbb{N}_d} {t^*_j}-T=0,\\
    &\text{(P4-5)}\quad \lambda^*_i \geq 0 \,\,\,\forall\, i \in \mathbb{N}_d\backslash\{k^*\},
\end{align}
where $(\mathbf{t}^*,\boldsymbol{\lambda}^*, \nu^*)$ satisfies KKT conditions. 
If KKT conditions are satisfied, strong duality holds ensuring $\mathbf{t}^*$ is the optimal solution of the primal problem. 

Partial derivative of $L(\mathbf{t},\boldsymbol{\lambda},\nu)$ about $t_x$ is represented by:
\begin{equation}\label{gradLaboutt}
    \begin{aligned}
        \frac{\partial{L}}{\partial{t_x}}&=
        \begin{cases}
            \lambda_x \nabla_{t_x}S_x(\mathbf{t})+\nu, \,\,& \mbox{if}\,\, x\neq k^*, \\
            \\
            (1-\sum\limits_{i \in \mathbb{N}_d \backslash \{k^*\}}{\lambda_i})\nabla_{t_{k^*}}S_{k^*}(\mathbf{t})+\nu, \,\,& \mbox{if}\,\, x = k^*.
        \end{cases}
    \end{aligned}
\end{equation}
In order to solve (P4-1), we put ${\partial{L}}/{\partial{t_x}}$ as $0$ and solve a simultaneous equation about $\lambda^*_x$ and $\nu^*$:
\begin{align}
\lambda^*_x &= \frac{\frac{1}{\nabla_{t_x} S_x(\mathbf{t}^*)}}{\sum\limits_{i \in \mathbb{N}_d}\frac{1}{\nabla_{t_i} S_i(\mathbf{t}^*)}}\label{lambda*}, \\
\nu^* &= -\frac{1}{\sum\limits_{i\in\mathbb{N}_d}\frac{1}{\nabla_{t_i} S_i(\mathbf{t}^*)}}, \label{nu*}
\end{align}
where $x \in \mathbb{N}_d\backslash\{k^*\}$, without consideration of division by zero due to (\ref{Sprime}). 
(P4-5) is also guaranteed by (\ref{Sprime}) and (\ref{lambda*}).
Since (\ref{lambda*}) stands for $\lambda^*_i \neq 0$, (P4-2) is simplified as:
\begin{align} \label{KKTresult}
    S_i(\mathbf{t}^*)-S_{k^*}(\mathbf{t}^*)=0 \,\,\,\forall\, i \in \mathbb{N}_d\backslash\{k^*\},
\end{align}
which is also satisfying (P4-3). 
In other words, (\ref{KKTresult}) is a sufficient condition for both (P4-2) and (P4-3). 
The KKT conditions ensure that obtained solution minimizes the primal problem, not the the solution exists. 
Referring to \cite{Chung23}, convergence of time slot allocation algorithm regardless of $\epsilon_1$ value is proven. 
By setting $\epsilon_1 \rightarrow 0, \epsilon_2 \rightarrow 0$, difference between maximum and minimum bits/Hz also converges to zero by the squeeze theorem with constant time slot summation. 
Therefore, we ensure that $\mathbf{t}^*$ satisfying (\ref{KKTresult}) and (P4-4) exists. 

Finally, $(\mathbf{t^*}, \boldsymbol{\lambda}^*,\nu^*)$ is the solution of KKT conditions (P4), which is also the optimal solution of the primal problem (P3) and the dual problem (\ref{dualproblem}). 
KKT conditions about $\mathbf{t}^*$ are (P4-4) and (\ref{KKTresult}), which is the requirements to satisfy (P3).
In addition, $R_k(\mathbf{t})$ is $-S_k(\mathbf{t})$.
Thus, we conclude that when $R_1(\mathbf{t}^*)=R_2(\mathbf{t}^*)=...=R_{N_d}(\mathbf{t}^*)$, $\mathbf{t}^*$ maximizes the minimum of $R_k(\mathbf{t})$ where $t_1+t_2+...+t_{N_d}=T$.

\end{proof}

\begin{algorithm}
\caption{GMGA for $\mathbf{c}^*$}\label{GMGA_alg}
\begin{algorithmic}[1] 
\STATE \textbf{main}
\STATE $p_1 = \text{direct connect}$
\STATE $p_i = \text{random connect} \quad \forall i \in [2,n_{origin}]$
\STATE $(\mathbf{R},\mathbf{t})= \text{evaluation}(\mathbf{p},\mathbf{t}_{\text{uniform}})$
\STATE $(\mathbf{p},\mathbf{R},\mathbf{t})= \text{selection}(\mathbf{p},\mathbf{R},\mathbf{t})$
\WHILE {$\mathbf{R}_{\text{max}} \text{ changes for last stop\_period}$}
\STATE $\mu_{i,j}=\frac{t_i \log_{2}(1+\frac{A_{i,j}}{t_i})}{\sum_{j \in \mathbf{N}}{t_i \log_{2}(1+\frac{A_{i,j}}{t_i})}}$
\STATE $o_{1,2,...,n_{\text{parent}}}=\mathbf{p}$
\STATE $o_{n_{\text{parent}}+1,...,n_{\text{offspring}}}=\text{mutation}(\text{crossover}(\mathbf{p}),\boldsymbol{\mu})$
\STATE $(\mathbf{R},\mathbf{t})=\text{evaluation}(\mathbf{o},\mathbf{t})$
\STATE $(\mathbf{p},\mathbf{R},\mathbf{t})=\text{selection}(\mathbf{o},\mathbf{R},\mathbf{t})$
\ENDWHILE
\STATE {$\mathbf{c}^* = p_1$}
\STATE \textbf{end main}
\\\hrulefill
\STATE \textbf{function} evaluation$(\mathbf{g}, \mathbf{t})$
\FOR {$i=1:\text{end}$}
\STATE $(R_i,t_i)=\text{time\_slot\_allocation\_alg}(g_i,\text{ts\_init}=t_i)$
\ENDFOR
\RETURN {$(\mathbf{R},\mathbf{t})$}
\STATE \textbf{end function}
\\\hrulefill
\STATE \textbf{function} selection$(\mathbf{g},\mathbf{R},\mathbf{t})$
\STATE $(\mathbf{R},\mathbf{I})=\text{sort}(\mathbf{R},\text{`descend'})$
\STATE rearrange $\mathbf{g},\,\mathbf{t}$ by sorted\_index, $\mathbf{I}$
\STATE $\mathbf{p}=g_{1,2,...,n_{\text{parent}}}$
\STATE $\mathbf{R}=R_{1,2,...,n_{\text{parent}}}$
\STATE $\mathbf{t}=t_{1,2,...,n_{\text{parent}}}$
\RETURN {$(\mathbf{p},\mathbf{R},\mathbf{t})$}
\STATE \textbf{end function}
\\\hrulefill
\STATE \textbf{function} crossover$(\mathbf{p})$
\FOR {$i=1:n_{\text{offspring}}-n_{\text{parent}}$}
\STATE $r_1,r_2=\text{random\_int}(1, n_{\text{parent}})$
\STATE $\text{div1}=\text{random\_int}(1, N_d)$
\STATE $\text{div2}=\text{random\_int}(\text{div1}, N_d)$
\STATE $g_i=\text{join}(p_{r_1}[1:\text{div1}],\,p_{r_2}[\text{div1}+1:\text{div2}],$\\
\STATE $\qquad \qquad p_{r_1}[\text{div2}+1:\text{end}])$
\ENDFOR
\RETURN {$\mathbf{g}$}
\STATE \textbf{end function}
\\\hrulefill
\STATE \textbf{function} mutation$(\mathbf{p},\boldsymbol{\mu})$
\FOR {$i=1:n_{\text{offspring}}-n_{\text{parent}}$}
\FOR {$j=1:N_d$}
\IF {random(0, 1) < $p_m$}
\STATE $g_i=$ mutate node $j$ of $p_i$ to uplink at \textit{$k$}th node with a probability of $\mu_{j,k}$.
\ENDIF
\ENDFOR
\ENDFOR
\RETURN {$\mathbf{g}$}
\STATE \textbf{end function}
\end{algorithmic}
\end{algorithm}

\subsection{Guided-Mutation Genetic Algorithm Based Topology Algorithm}

In this subsection, we propose a GMGA in order to find a sub-optimal solution of $\mathbf{c}$ in our formulated problem, which modulates mutation rate regarding the cost of the links between nodes.

Finding a relaying topology based on exhaustive search method requires searching from $(N_{d}+1)^{(N_d-1)}$ possible candidates, making it an NP-hard problem and implausible for real-world application.
Despite previous attempts to address this issue, practical application remains challenging due to extended computation times and comparably low minimum bits/Hz than optimal solution.
For instance, the approach using a variational autoencoder (VAE) in \cite{Chung23} aimed to train layers by devising loss function evaluation method called PT-EVM and minimizing it to discover sub-optimal solution of $\mathbf{c}$. 
% R2.5
However, this method faced challenges in practical implementation due to time-consuming nature of neural networks in VAE scheme.

In parallel with these challenges, in prior approach, GA is also commonly used as a solver for NP-hard problems in a wide area. 
Such application is also applied to our formulated problem, since our system model takes the form of a tree structure, which can be divided into branches that is easy to evaluate whether it helps overall performance or not.
By defining genes of chromosome as the index of the parent node, GA can handle our formulated problem by selecting the beneficial branches and inheriting them to the next generation.

With these advantages, GA operates by selecting the best-performing candidates from a generation and applying crossover and mutation to pass on their traits to the next generations. 
Throughout this process, GA compares and refines high-score expected candidates rather than evaluating every possible candidate. 
Since our formulated problem has requirements including all nodes being connected to a sink and having only one parent node through an uplink, we add some details in our proposed GA scheme.  
Detailed explanations of four components (evaluation, selection, crossover, and mutation) of our scheme are as follows: 

First, in the evaluation part, the candidates are evaluated based on their measured performances.
The evaluation part called at line 4 of Alg. \ref{GMGA_alg} which is defined from line 15 to 20 signifies the evaluation for the first generation, while line 10 represents the evaluation for the other generations. 
Since our objective is to maximize the $R_{\text{min}}$, we designate the $R_{\text{min}}$ value obtained through time slot allocation algorithm as a score of a candidate. 

Second, in the selection part called at line 5 and 11 of Alg. \ref{GMGA_alg} which is defined from line 21 to 28, all candidates are sorted in descending order, and a specific number of high-scored candidates are then selected as parents for the next generation.

Third, in the crossover part called at line 9 of Alg. \ref{GMGA_alg} which is defined from line 29 to 38, the chromosomes of two parent candidates are combined, creating a new offspring chromosome with new attributes. 
Crossover occurs at $n_c$ points, and in the Alg. \ref{GMGA_alg}, we provide an example for the case where $n_c=2$. 

Finally, in the mutation part called at line 9 of Alg. \ref{GMGA_alg} which is defined from line 39 to 48, the genes of the individual are mutated into random values under a mutation probability $p_m$, in order to explore other possibilities which cannot be induced from crossover operations.  
After completing the crossover or mutation part, a tree validation test is conducted at the end of each process to satisfy the constraints (\ref{c2}) and (\ref{c3}) in (P1). 
In addition, we include elitism in our algorithm to ensure the enhancement of performance as described in line 8 of Alg. \ref{GMGA_alg}.

\begin{figure*}
	\centering
	\includegraphics[width=1.0\linewidth]{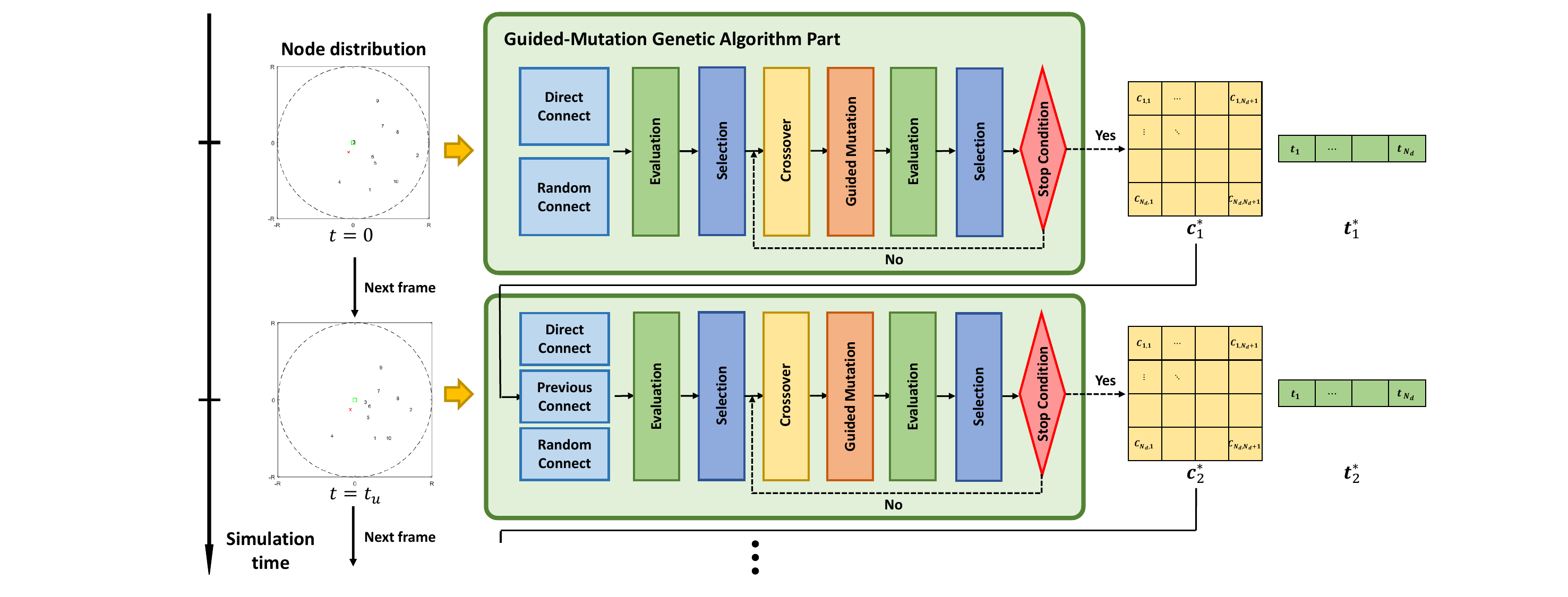} 
	\caption{Structure flow chart of the mobility-aware iterative relaying topology algorithm to find $\mathbf{c^*}$ and $\mathbf{t^*}$ throughout the simulation.}
	\label{struct}
 \vspace{-0.5 cm}
 \end{figure*}

Above mentioned descriptions are the main structure of GA for application to our formulated problem. 
The basic form of GA uses the uniform mutation rate, which means that when a node A is selected to be mutated, the probability of another node being selected as an uplink for node A is uniform. 
However, the uniform mutation rate leads to an excess of low-performance candidates in our formulated problem, such as nodes connecting directly to the sink without sufficient energy or to the distant nodes. 
Furthermore, uniform mutation rate results in some performance related problems, including an increment of candidates which will not be selected in the evaluation part, reducing the variation between parents and offsprings. 
The issue of topology remaining unchanged throughout generations mentioned above can lead to very early stopping under the mistaken decision that the algorithm has reached an optimal topology, resulting in premature outcomes.

In our proposed GMGA method, the mutation rate is adjusted to induce mutations toward more reasonable node connections. 
Specifically, it is adjusted to be inversely proportional to the link cost so that the probability of a specific node being chosen as an uplink is varied. 
Through this adjustment, reasonable mutations such as mutating towards adjacent nodes are expected.
The time slot used for calculating link cost is derived from the time slots of parents, since it is closer to the optimal time slot than time slot obtained by uniformly dividing the time frame. 
Detailed explanation of mutation rate $\mu_{i,j}$, which is the conditional probability of uplink of node i mutating to node j given that node i mutates, is as follows:
\begin{align}
    \mu_{i,j}&=P(\text{node i uplinks to node j | node i mutates})\\
            &=\frac{t_i\log_{2}{(1+\frac{A_{i,j}}{t_i})}}{\sum_{j \in \mathbb{N}}{t_i\log_{2}{(1+\frac{A_{i,j}}{t_i})}}}
\end{align}
where $\mathbf{t}$ is the time slot allocation of parents.

The guided-mutation results in less computation time than basic GA with uniform mutation rate.  
Besides, in order to boost our scheme more for the real-world application, we devise several tunings as follows:

First, we include the case of connecting all nodes directly to the sink in the first generation, as described in line 2 of Alg. \ref{GMGA_alg}. 
One of the advantages of GA is the ability to include `promising' candidates in the learning process without modulation of the algorithm scheme, allowing them to compete with other candidates. 
Typical example of `promising' candidates is a case of direct connection of all nodes to the sink. 
If all nodes have sufficient energy, it is better to connect directly to the sink since relaying is detrimental. 
Therefore, in the first generation, we include a case of direct connection to the sink, considering scenarios where there are enough PBs or nodes densely clustered around PBs. 
While this tuning could be perceived as disrupting exploration, the possibility of a random tree showing higher performance than direct connect is rare. 
Moreover, since only one candidate among several candidates in the first generation is modified, it has minimal impact on the exploration.

Second, we start time slot allocation from the time slot allocation result of parents except first generation in order to calculate time slot of offsprings faster, since the topologies of offsprings are similar to the topologies of parents. 
This approach is experimentally proven to be more time-efficient than starting calculations from a uniform time slot. 
Additionally, scores of all topologies generated during the learning process are stored, so that if a score calculated in a past generation is redundantly required in the current generation, the stored value is retrieved without repetitive calculations. 
Since time slot allocation algorithm and calculation of rate are the most time consuming part of our algorithm, the storage of calculation is effective in reducing the computation time.  

Finally, in the evaluation part, using a time slot allocation algorithm with small $\epsilon_1$ requires more iterations to minimize the difference between the maximum and minimum $R_k(\mathbf{c},\mathbf{t})$. 
Meanwhile, in the evaluation process, precise $R_k(\mathbf{c},\mathbf{t})$ values are not required.
Instead, our objective is to compare which topology is better. 
Therefore, we define $\epsilon'_1$ which is bigger than $\epsilon_1$ for the time slot allocation in the evaluation part, allowing faster comparison and tolerant convergence. 
Note that the $\epsilon'_1$ value is sufficiently small, so that we determine that time slot allocation is done for comparison purpose when the difference between the maximum and minimum values of $R_k(\mathbf{c},\mathbf{t})$ fell below $\epsilon'_1$.

\subsection{Mobility-aware Iterative Relaying Topology Algorithm}
In this subsection, we describe a relaying topology algorithm regarding mobile nodes.

Previously investigated relaying topology algorithms are hard to apply in the real-world which involves mobile nodes. 
Some methodologies with a rapid computation result in insufficient capacity for transmit their data due to an inefficient topology.  
Conversely, some methodologies showing good performance fail to ascertain a topology before the subsequent simulation frame. 
This trade-off between computational speed and performance impedes the adoption of relaying topologies in the real-world with mobile nodes since the nodes unable to operate in the calculated topology which is sub-optimal for the current time frame.
However, GA has a simple structure which facilitates calculations and ensures performance improvements over generations. 
Unlike other algorithms, GA can start calculations from the previously calculated result while avoiding the redundant calculation. 
Thus, by using these advantages, we utilize GA which has relatively fast calculation and good performance. 

As described above, the merits of GA are incorporating promising candidates into the learning process and refining the candidates to find a sub-optimal solution of $\mathbf{c}$, iteratively.
Thus, since initiating the calculation with a relaying topology which is close to the optimum enhances the performance, we assign the final relaying topology results from the previous frame to the first generation of current frame for the mobility simulation. 
The reason of using the relaying topology result from the previous frame is as follows.

While nodes move during a unit time interval, relocated nodes are akin to the pre-movement state. 
Additionally, the topology of the system model only depends on a distribution of the nodes if other conditions are unchanged during the unit time interval.  
Therefore, the changes in relaying topology across consecutive simulation frames are minimal.
By incorporating of the final relaying topology resulted from the previous frame into the first generation of the current frame, it provides more reasonable starting point than a random tree. 
Note that the incorporation of final results from the previous frame has a minimal impact on exploration, similar to the above mentioned incorporation of the direct connection case.

Including all of our considerations, flow chart of our mobility-aware iterative relaying topology algorithm is depicted in Fig. \ref{struct}.

%%%%%%%%%%%%%%%%%%%%%%%%%%%%%%%%%%%%%%%%%%%%%%%%%%%%%%%%%%%%%%%%%%%%%%%%%%%%%
\section{Performance Analysis and Discussions}
\subsection{Test Environment and Parameter Configuration}

\begin{table*}[]
\centering
\renewcommand{\arraystretch}{1.2}
\small
\begin{tabular}{llll}
\Xhline{3\arrayrulewidth}
Category & Parameter & Definition & Value \\ 
\Xhline{3\arrayrulewidth}
\multirow{2}{*}{Environment} & $R$ & Radius of circle that nodes are distributed & 0.5 km \\
 & $P_b$ & Power of power beacon & 1 W \\ \hline
\multirow{5}{*}{TDMA system} & $T$ & Time frame & 100 msec \\
 & $|h_{k,n}|^2$ & Channel between \textit{$k$}th node and \textit{$n$}th node & $\sim$ Exp(1) \\
 & $\alpha$ & Path loss exponent & 3 \\
 & $BW$ & Bandwidth & 125 kHz \\
 & $NF$ & Noise figure & 6 dB \\ \hline
\multirow{5}{*}{\begin{tabular}[c]{@{}l@{}}Genetic \\ algorithm\end{tabular}} & $n_{origin}$ & The number of candidates in the first generation & 5 \\
 & $n_{parent}$ & The number of parents & 5 \\
 & $n_{offspring}$ & The number of offsprings & 50 \\
 & $n_c$ & The number of crossover points & 2 \\
 & $p_m$ & Mutation probability & 0.05 \\ \hline
\multirow{3}{*}{\begin{tabular}[c]{@{}l@{}}Time slot\\ allocation\\ algorithm\end{tabular}} & $\epsilon_1$ & Upper limit of rate budget difference (for rate budget calculation) & $10^{-6}$ bits/Hz\\
 & $\epsilon'_1$ & Upper limit of rate budget difference (for evaluation part) & $10^{-3}$ bits/Hz\\
 & $\epsilon_2$ & Minimum allocatable time slot & $10^{-7}$ sec\\ \hline
\multirow{3}{*}{Mobility} & $v_c$ & Speed of each node & 6.42 m/s \\
 & $T_s$ & Total mobility simulation time & 180 sec \\
 & $T_u$ & Unit mobility simulation time & 20 sec \\
\Xhline{3\arrayrulewidth}
\end{tabular}
\caption{Experimental parameters used in simulation.}
\label{exp_param}
\end{table*}

In this subsection, we define the system model parameters for evaluating the performance of our proposed scheme.

In our system model, nodes are distributed within a circle with a radius $R$ of 0.5 km. 
Each node transmits data under a TDMA system with a time frame $T$ of 100 milliseconds. 
For the energy harvesting model, the transmit power of the Power Beacon (PB) is set to 1W. 
The energy harvesting model is assumed as a linear model, where the energy of node is 0.7 times the harvested energy sent from the PBs. 
The wireless fading channel, represented by $|h_{k,n}|^2$, follows an exponential random variable with a unit mean. 
The path loss exponent $\alpha$ is 3, the bandwidth is 125 kHz, and the noise figure $NF$ is 6dB in our system model, respectively. 
The Noise power $N$ is calculated under $N=-174+NF+10\log_{10}{BW}$.

Next, we set the hyper-parameters for our GA as follows. 
First, the size of the first generation $n_{\text{origin}}$ and the size of the parent generation $n_{\text{parent}}$ are both 5. 
The size of the offspring generation $n_{\text{offspring}}$ is set to 50. 
Next, the number of crossover points is assumed to be $n_c= 2$, and the mutation probability to be $p_m = 0.05$. 

In the evaluation part, the tolerant value of $\epsilon_1$ which is $\epsilon'_1$ is set to $10^{-3}$ bits/Hz. 
For the time slot allocation algorithm which is used to calculate the exact performance of an algorithm, $\epsilon_1$ is set to $10^{-6}$ bits/Hz, and $\epsilon_2$ is set to $10^{-7}$ sec. 
Finally, we set the stop condition to halt the GA if $R_{\text{min}}$ remains unchanged for $200/N_{d}-4$ generations.

For the mobility-aware system model, parameters are set as follows. 
First, the speed of each node is set to $6.42$ m/s, which is the average speed of vehicles in Seoul, Korea \cite{Seoul22}. 
Second, the staying time of RWM is set to zero in order to assume worst case regarding mobility. 
If a node arrives at its destination, the node waits until the termination of current frame and sets the next random destination. 
Finally, the total simulation time $T_s$ is 3 minutes, and the unit simulation time $T_u$ is 20 seconds.
Thus, the number of simulation frame for the given number of nodes, PBs and scheme is 10. 
We present the values and definitions of the experimental parameters in Table \ref{exp_param}.

To compare the performance of our GMGA with the performance of other schemes, we choose 4 different schemes and optimal solution (if applicable) case along with our proposed scheme:

\begin{itemize}
    \item \textbf{Optimal solution (Opt)} which exhaustively compares the performance of all possible trees in order to find the optimal solution, and allocates time slot using IB time slot allocation algorithm.  
    \item \textbf{Direct connect scheme (Dir)} which connects every node to the sink and allocates time slot using time slot allocation algorithm. 
    \item \textbf{MST scheme (MST)} which makes a Minimum Spanning Tree (MST) based on the link cost denoted as $1/\left( t_i \log _{2} (1+\Gamma_{i,j})\right)$ starting from the sink with uniform time slot allocation, and allocates time slot using time slot allocation algorithm.  
    \item \textbf{Greedy scheme (Greedy)} which establishes initial connections by linking all nodes directly to the sink and iteratively selects a node $i$ randomly and connects it to the node with the highest achievable rate, determined by $\min(t_i \log_{2} (1+\Gamma_{i,j}), \mathbf{B}_{j})$ $\forall j \in \mathbb{N}_d$, and allocates time slot using time slot allocation algorithm. 
    \item \textbf{VAE scheme (VAE)} which decides the topology with VAE and PT-EVM scheme \cite{Chung23}, and allocates time slot using time slot allocation algorithm. 
    \item \textbf{Proposed scheme (Prop)} which decides the topology with GMGA, and allocates time slot using time slot allocation algorithm. 
    \begin{comment}
    \end{comment}
\end{itemize}

All considered schemes are implemented on Matlab R2023a, using a computer with an Intel Core i7-12700 and 32GB of memory, without GPU acceleration.

In an environment where nodes are mobile, it is necessary to compute a relaying topology with a proper computation time, so that it can reflect the current distribution of the nodes accordingly.
Thus, not only regarding the importance of computation time presented above, but also examining the quantitative comparison between the existing schemes and our proposed scheme matters, we focus on performance metrics rate budget and computation time in this paper.

\begin{figure}
	\centering
 \mbox{
 \subfigure[time = 0 sec]{\label{0sec}
	\includegraphics[width=0.3\linewidth]{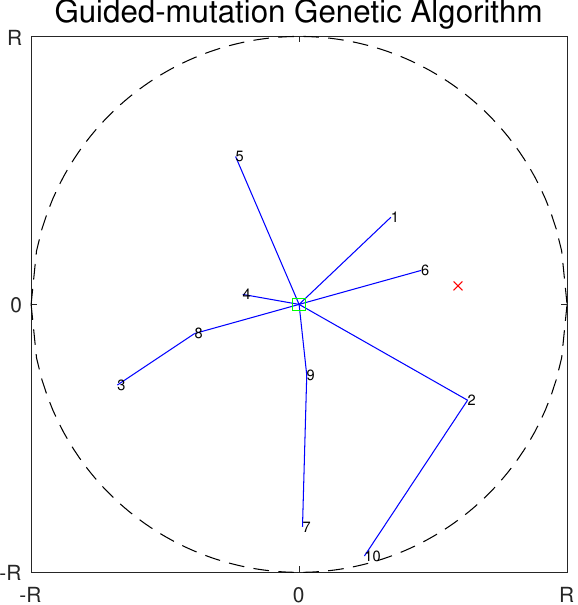} 
 }
  \subfigure[time = 20 sec]{\label{20sec}
 \includegraphics[width=0.3\linewidth]{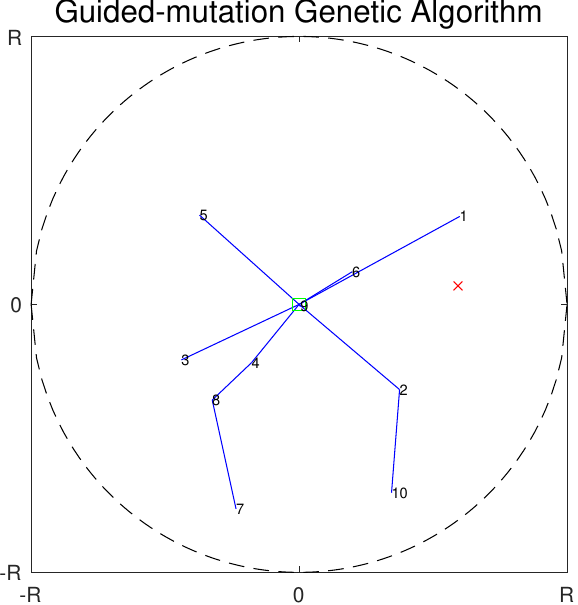} 
 }
  \subfigure[time = 40 sec]{\label{40sec}
 \includegraphics[width=0.3\linewidth]{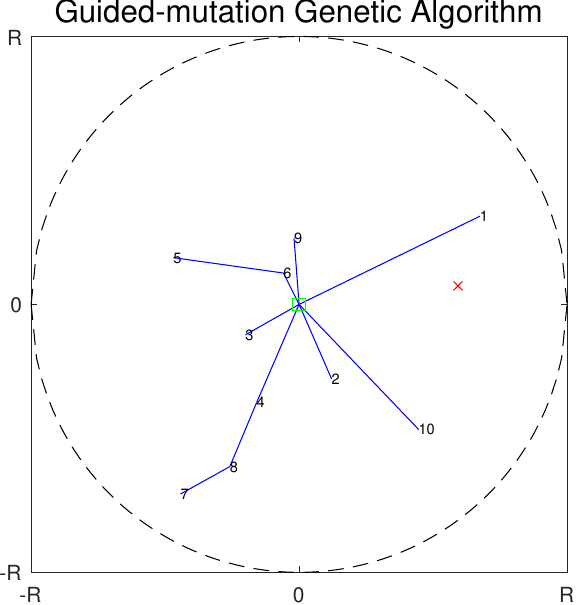} 
 }
    }
	\caption{Change in relaying topology of mobility simulation over time for $N_{d}=10$, $N_{b}=1$ example.}
	\label{mobility_simulation}
  \vspace{-0.5 cm}
\end{figure}

\begin{figure}[t]
	\centering
	\includegraphics[width=0.9\linewidth]{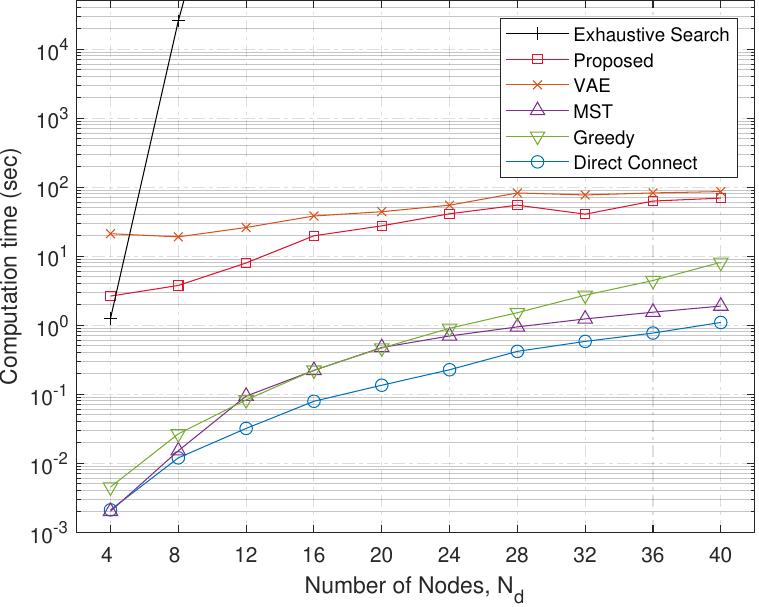} 
	\caption{Computation time of 6 different schemes with respect to the number of nodes, $N_d$.}
	\label{Ex3}
  \vspace{-0.5cm}
\end{figure}

\begin{figure*}
	\centering
	\mbox{\subfigure[$N_d=5$]{\label{Ex1-1} \includegraphics[width=0.33\linewidth]{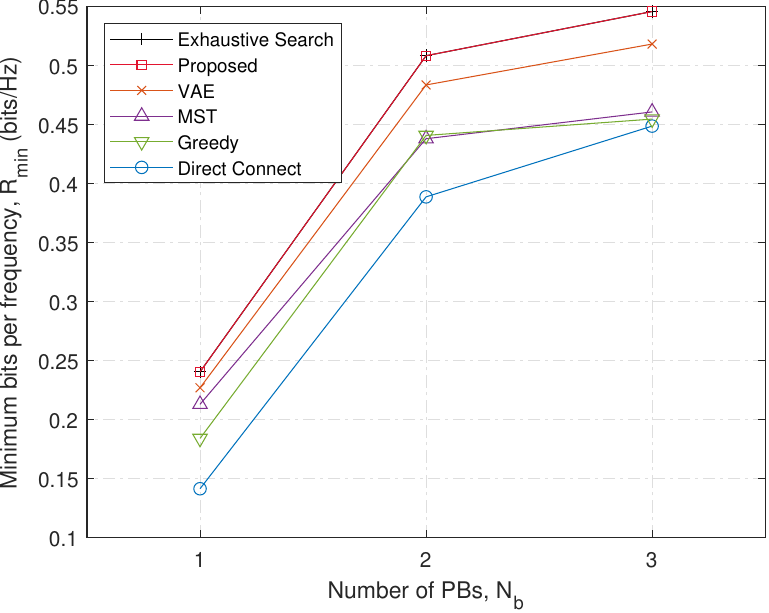}
	}
	\subfigure[$N_d=6$]{\label{Ex1-2} \includegraphics[width=0.33\linewidth]{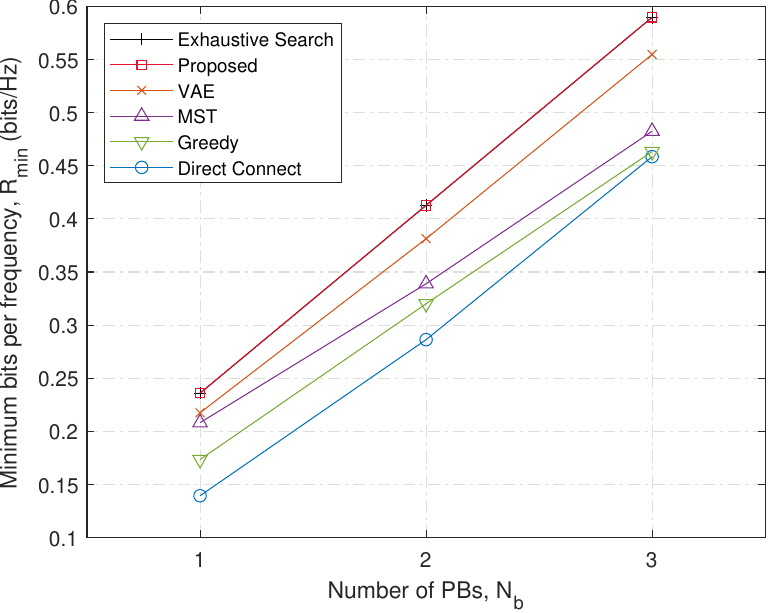}
    }
    \subfigure[$N_d=7$]{\label{Ex1-3} \includegraphics[width=0.33\linewidth]{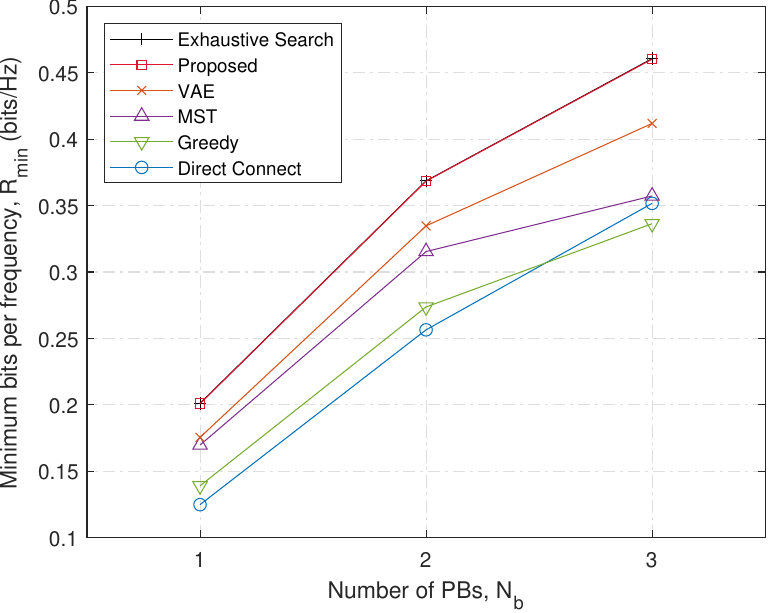}
	}
    }
	\caption{\mbox{Stationary performance comparison with respect to the number of PBs, $N_b$ in $N_d\in\{5,6,7\}$}}
	\label{Ex1}
  \vspace{-0.3cm}
\end{figure*}

\begin{figure*}
	\centering
	\mbox{\subfigure[$N_d=10$]{\label{Ex2-1} \includegraphics[width=0.33\linewidth]{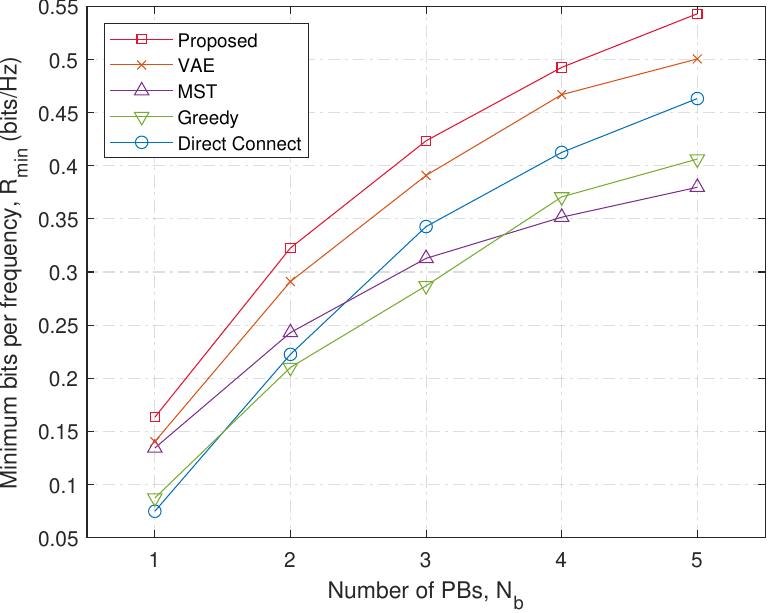}
	}
	\subfigure[$N_d=20$]{\label{Ex2-2} \includegraphics[width=0.33\linewidth]{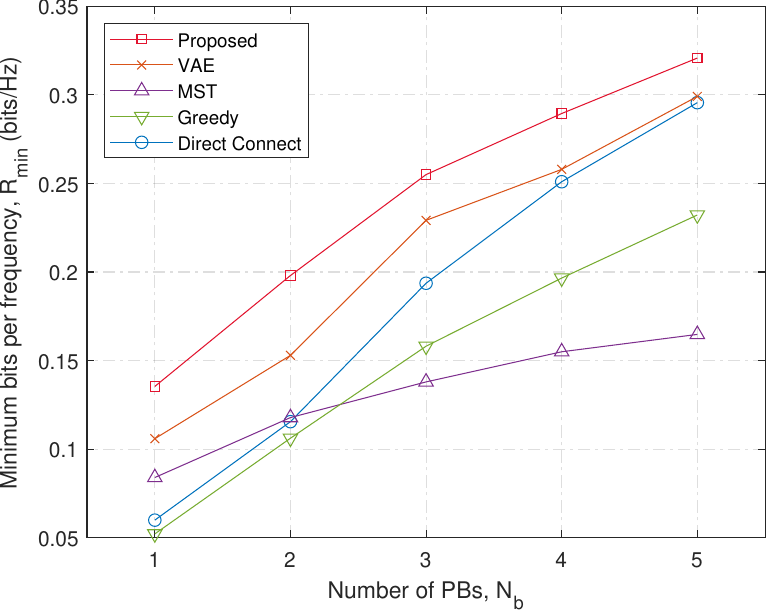}
    }
    \subfigure[$N_d=30$]{\label{Ex2-3} \includegraphics[width=0.33\linewidth]{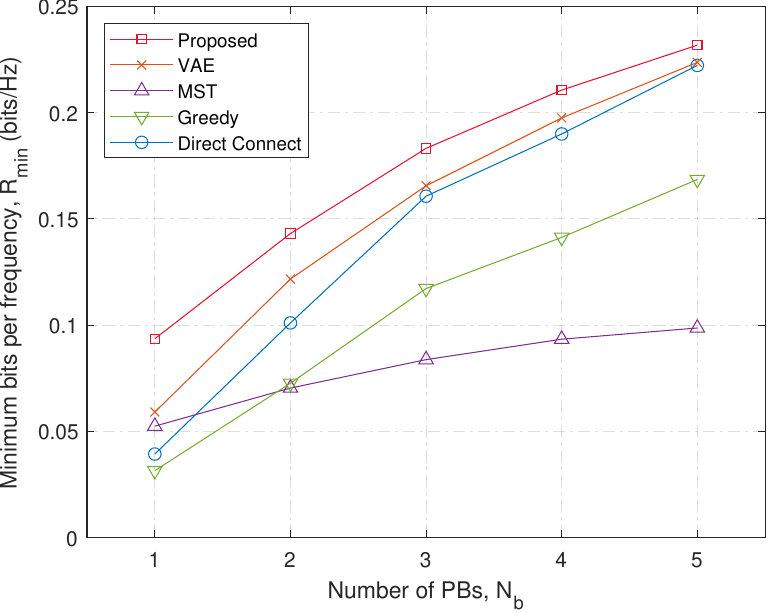}
	}
    }
	\caption{\mbox{Stationary performance comparison with respect to the number of PBs, $N_b$ in $N_d\in\{10,20,30\}$}}
	\label{Ex2}
  \vspace{-0.3cm}
\end{figure*}

\begin{figure*}
	\centering
	\mbox{\subfigure[$N_d=5$]{\label{Ex4-1} \includegraphics[width=0.33\linewidth]{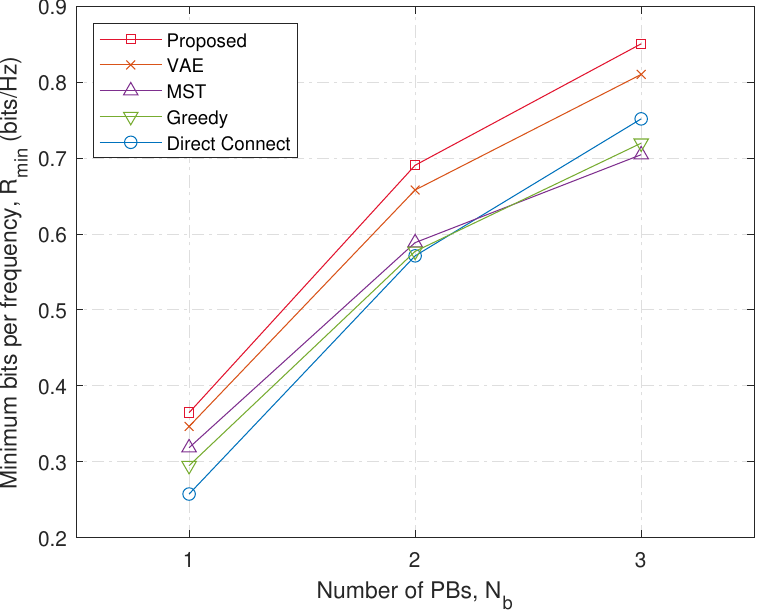}
	}
	\subfigure[$N_d=10$]{\label{Ex4-2} \includegraphics[width=0.33\linewidth]{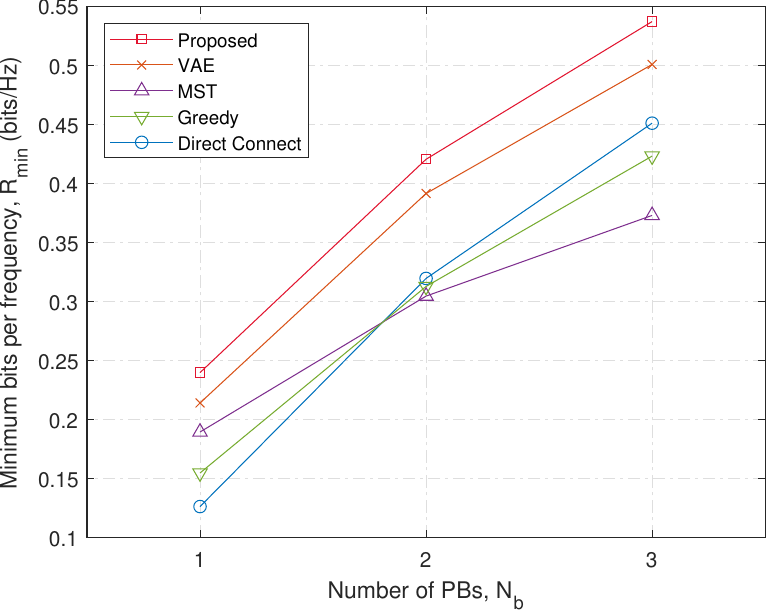}
    }
    \subfigure[$N_d=20$]{\label{Ex4-3} \includegraphics[width=0.33\linewidth]{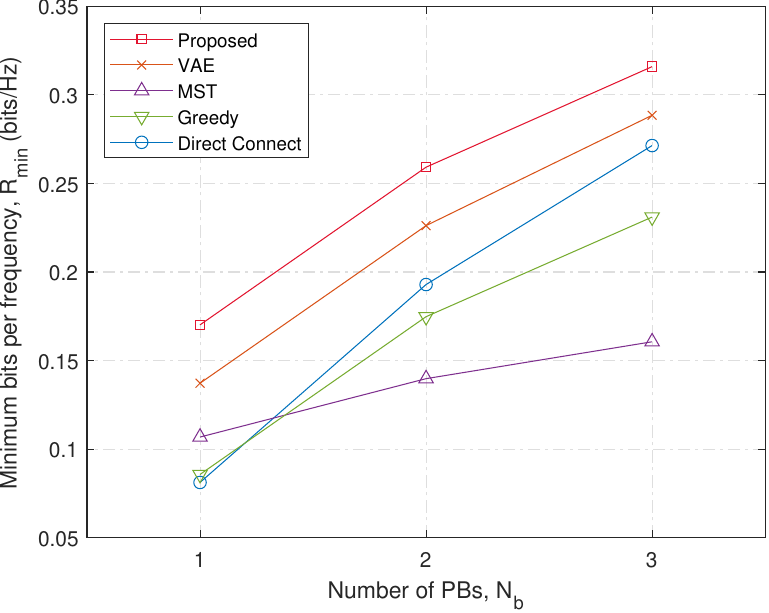}
	}
    }
	\caption{\mbox{Mobility performance comparison with respect to the number of PBs, $N_b$ in $N_d\in\{5,10,20\}$}}
	\label{Ex4}
  \vspace{-0.3cm}
\end{figure*}

\begin{figure*}
	\centering
	\mbox{\subfigure[$N_b=1$]{\label{Ex5-1} \includegraphics[width=0.33\linewidth]{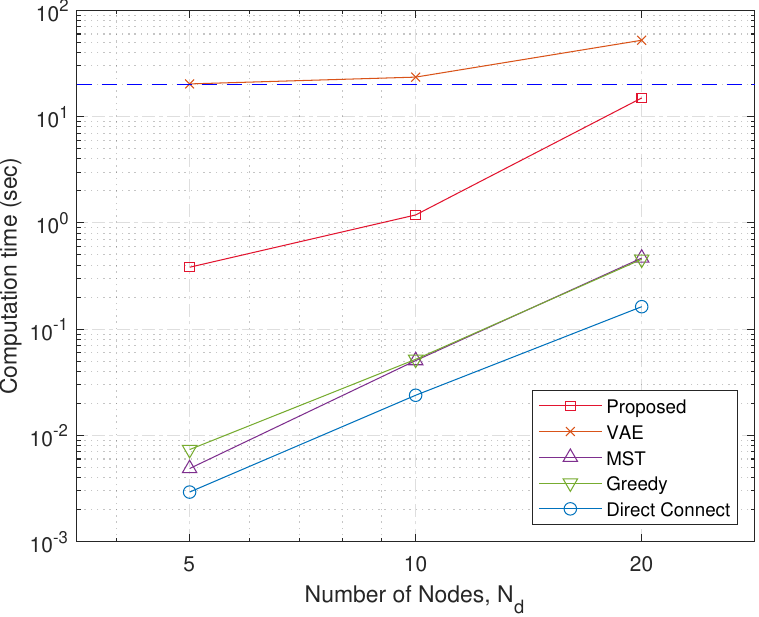}
	}
	\subfigure[$N_b=2$]{\label{Ex5-2} \includegraphics[width=0.33\linewidth]{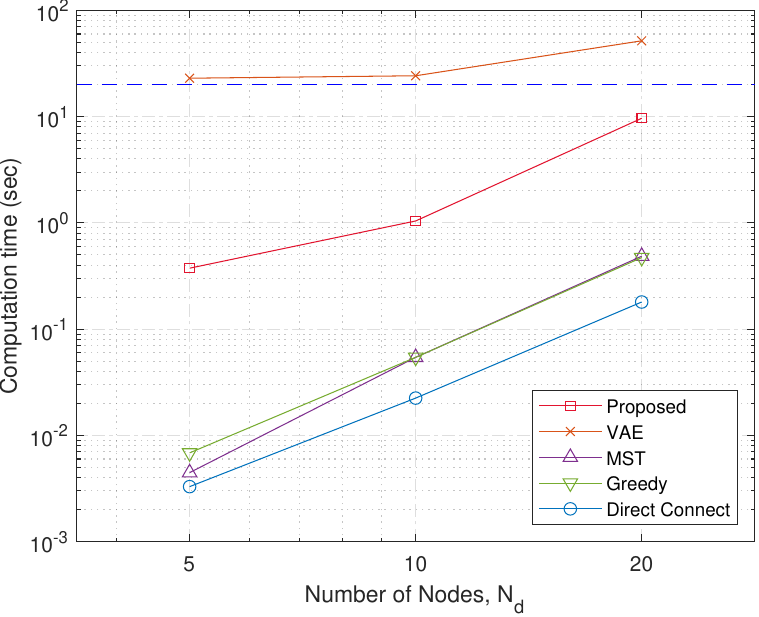}
    }
    \subfigure[$N_b=3$]{\label{Ex5-3} \includegraphics[width=0.33\linewidth]{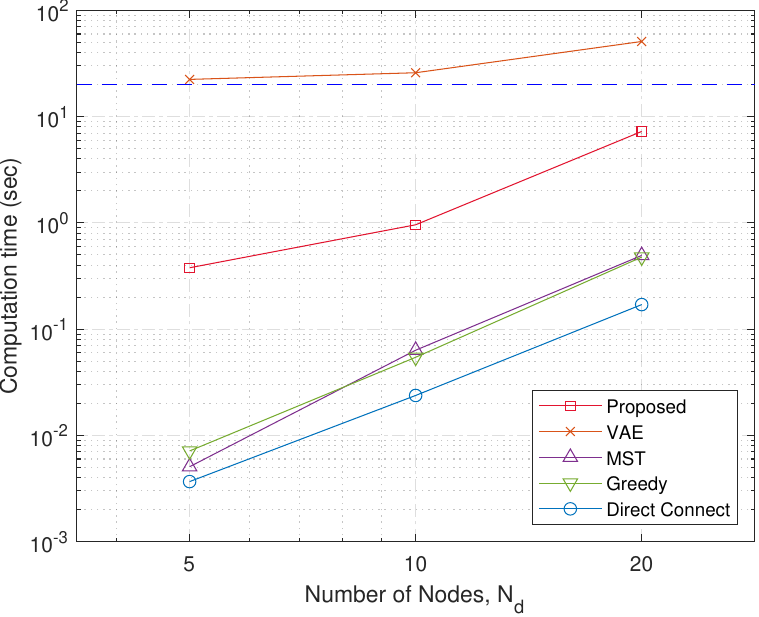}
	}
    }
	\caption{\mbox{Mobility performance comparison with respect to the number of nodes, $N_d$ in $N_b\in\{1,2,3\}$}}
	\label{Ex5}
  \vspace{-0.3cm}
\end{figure*}

\subsection{Experiment Analysis and Discussion}

Fig. \ref{mobility_simulation} shows the change in relaying topology obtained by proposed scheme as all nodes are moving under RWM. 
Each sub-figure \ref{0sec}, \ref{20sec}, \ref{40sec} shows the sample distribution of 10 nodes when the simulation time is 0, 20, 40 seconds, respectively. 
The relaying topology of each distribution is changing over time for high $R_\text{min}$ of changing distribution since relaying topology is dependent on the distribution of nodes.

Fig. \ref{Ex3} shows the computation time of considered schemes with respect to the number of nodes from $N_d=4$ to $40$ with an interval of 4. 
Note that the maximum number of nodes is set to 40 in this experimental configuration, with reference to \cite{Birtane24}, which set the maximum number of nodes in a WSN IoT network at 32, and \cite{Li20}, which set the maximum number of nodes in an IoT network at 30 including malicious nodes.
Each data point of Fig. \ref{Ex3} is an averaged computation time results of 30 random distributions of nodes and PBs. 
The computation time of exhaustive search method is excessive for real-world application, suggesting the need for an alternative scheme.
Also, since other deep learning techniques have characteristics of neural networks, adding VAE as a comparative group scheme could be considered as a representative for comparisons over the other deep learning schemes.

In Fig. \ref{Ex3}, the computation time increases as the number of nodes increases due to the increased complexity of relaying topology problem. 
However, it can be seen from the computation time results for the proposed scheme that there are some cases where the computation time decreases even as the number of nodes increases, especially around in a data point with 32 nodes.
This fluctuation results from a statistical phenomenon due to the finite sample size, along with the randomness characteristics of the GA. 
Nevertheless, the tendency for computation time to increase as the number of nodes increases can still be confirmed in Fig. \ref{Ex3}.

Additionally, our proposed GMGA takes less computation time than the VAE scheme while takes more computation time than the other conventional schemes. 
Note that the computation time for VAE includes the time spent in the training phase along with the inference phase, since the VAE scheme proposed in \cite{Chung23} has a structure where training phase and inference phase are inseparable, and requires specialized training for each topology, followed by its respective inference.
Unlike VAE scheme using back propagation for learning, the GMGA choose the best relaying topology among candidates, requiring less computation time than VAE scheme.
In addition, GMGA avoids redundant scoring of relaying topology, decreasing the usage of laborious time slot allocation algorithm. 
Note that the difference in computation time between other conventional schemes and our proposed scheme becomes smaller as the number of mobile nodes increases. 

Fig. \ref{Ex1-1} to Fig. \ref{Ex1-3} show the $R_{\text{min}}$ of considered schemes for a stationary simulation with respect to $N_b \in \{1,2,3\}$ when $N_d \in \{5,6,7\}$.
Each data point of Fig. \ref{Ex1-1} to Fig. \ref{Ex1-3} is an averaged performance results of 30 random distributions of nodes and PBs. 
Note that value of $R_k(\mathbf{c},\mathbf{t})$ can be transformed into a rate by multiplying $BW/T$ which is constant in our test environment. 
In addition, $R_k(\mathbf{c},\mathbf{t})$ is calculated by the IB time slot allocation algorithm whose $\epsilon_1=10^{-6}$ so that the difference between the maximum and minimum $R_k(\mathbf{c},\mathbf{t})$ is negligible. 
Our proposed scheme shows the superior performance among considered schemes for every condition. 
Moreover, the results of our proposed scheme overlaps the results of exhaustive search scheme, which means that our proposed scheme finds the optimal relaying topology when $N_b \in \{1,2,3\}$, $N_d \in \{5,6,7\}$. 
Additionally, as the power beacon is added, the overall received power at each node increases, and therefore it can be seen that the minimum bits/Hz increases accordingly across all considered schemes.

Fig. \ref{Ex2-1} to Fig. \ref{Ex2-3} show the $R_{\text{min}}$ of considered schemes with respect to $N_b \in \{1,2,3,4,5\}$ when $N_d \in \{10,20,30\}$. 
Each data point of Fig. \ref{Ex2-1} to Fig. \ref{Ex2-3} is also an averaged performance results of 30 random distributions of nodes and PBs. 
As in the previous case, our proposed scheme shows superior performance among the considered schemes for every condition. 
Similar to the results of previous $N_d \in \{5,6,7\}$ case, it can be observed that the rate budget of each node gradually increases as the received power of each node increases.
Note that simulation using exhaustive search method is not conducted due to excessive computation time.

As shown in Fig. \ref{Ex1} and Fig. \ref{Ex2}, our proposed scheme shows the superior performance among various $N_d$ from 5 to 30 and $N_b$ from 1 to 5 compared to the above mentioned conventional schemes, illustrating the scalability of our proposed method.
We suggest the reason of the superior performance as follows: 
The guided mutation guide the connection of node into more rational way.
The optimal relaying topology is likely to have connections from one node to another node where the cost is lower rather than higher. 
While connecting to a node with higher cost can be beneficial in a global aspect, it is rare for the final optimal result. 
Therefore, in the GMGA, the conditional probability $\mu_{i,j}$ given that a mutation occurs is likely to be higher than $1/N_d$, which is the conditional probability of randomly selecting a node. 
Let us assume that $i$th node connects to $j$th node which is an optimal connection. 
Since there is no way to know whether this connection is optimal during a process, mutations occasionally occur at $i$th node. 
For the uniform mutation rate, the conditional probability of the $i$th node connecting to a different node instead of the $j$th node due to a misjudgment is $1-1/N_d$. 
However, for the GMGA, the conditional probability of misjudgment is $1-\mu_{i,j}$, which is lower than $1-1/N_d$. 

Fig. \ref{Ex4-1} to Fig. \ref{Ex4-3} show the $R_\text{min}$ of considered schemes for a mobility simulation with respect to $N_b \in \{1,2,3\}$ when $N_d \in \{5,10,20\}$. 
Each data point of Fig. \ref{Ex4-1} to Fig. \ref{Ex4-3} is an averaged performance results of 30 random distributions of nodes and PBs during mobility simulation. 
In other words, each data point is an average of $30$ random distributions consisting of $10$ simulation frames, totaling $300$ simulation results. 
Fig. \ref{Ex4-1} to Fig. \ref{Ex4-3} show that our proposed scheme has a superior performance among considered schemes for every condition. 
Note that the number of nodes is selected as 5, 10, and 20 in order to achieve similar coverage to stationary simulations and to prevent excessive computation time while conducting simulations, which requires about ten times more calculations than stationary simulations.

Fig. \ref{Ex5-1} to Fig. \ref{Ex5-3} show the computation time of considered schemes for a mobility simulation with respect to $N_d \in \{5,10,20\}$ when $N_b \in \{1,2,3\}$. 
Each data point of Fig. \ref{Ex5-1} to Fig. \ref{Ex5-3} is also an averaged performance results of 30 random distributions of nodes and PBs during mobility simulation. 
Fig. \ref{Ex5-1} to Fig. \ref{Ex5-3} show that as the number of nodes increases, computation time increases due to the increased complexity of the problem. 
In addition, it shows that the computation time of our proposed algorithm is shorter than the VAE scheme and longer than the other schemes. 
The computation time of the direct connect scheme is shortest since it only calculates a single time slot allocation of a defined topology. 
On the other hand, the computation time of the VAE scheme is the longest since it calculates numerous weights by back propagation. 
The blue dashed line of Fig. \ref{Ex5-1} to Fig. \ref{Ex5-3} indicates the unit simulation time, specified as 20 seconds in this test environment. 
The result shows that the computation time of our proposed scheme is shorter than the unit simulation time, which means that calculation of the relaying topology of the current frame without accumulating calculation demands is possible.
Therefore, a mobility-aware iterative relaying topology algorithm is appropriate to find a sub-optimal relaying topology in the real-world. 

As shown in Fig. \ref{Ex4} and Fig. \ref{Ex5}, our proposed scheme shows the superior performance compared to the above mentioned conventional schemes.
The reason we suggest is an inheritance of a final result, while other conventional schemes are not considered to start calculations from a given topology. 
However, our proposed scheme can start calculation from the given topology which is the final topology of previous simulation frame in our consideration, enabling the calculation from an advantageous status. 
If starting computation of the relaying topology from a random distribution is considered, potential area of each node corresponds to a region with a radius of $R$, with an area of $R^{2}\pi$. 
On the other hand, if starting computation from the final result of the previous simulation frame is considered, maximum potential area of each node is $v_c T_u$ away from the previous location of node, with an area of $(v_c T_u)^2\pi$.
Note that $(v_c T_u)^2\pi$ is a maximum value considering a node which locates nearby the edge of simulation circle with a radius R. 
Therefore, compared to the random distribution, calculation starting from the previous simulation frame has a positional uncertainty of $({v_c T_u}/{R})^2$ which is about 6.59\% in our parameter configuration. 
Furthermore, since this value pertains to a single node, starting calculations from the final result of previous simulation frame significantly reduces uncertainty in the mobile IoT system with a large number of nodes.
Thus, fewer changes in the topology are expected than starting calculations from a random node distribution, leading to a superior performance and reduced computation time.

In addition, in Fig. \ref{Ex5}, our proposed scheme takes less computation time than the VAE scheme but takes more computation time than the other conventional schemes. 
We assume that the results related to computation time come from the simplicity of the GA. 
The GA consists of crossover, mutation, evaluation, and selection parts, and the most time-consuming part is the evaluation part since it includes the time slot allocation algorithm which involves recursive logarithmic operations. 
However, the time slot allocation algorithm becomes less complex since the adjustment of $\epsilon'_1$ enables the tolerant calculation.
Let us suppose that $\Delta$, which is the difference between the maximum rate budget and the minimum rate budget, decreases by half with every iteration \cite{Chung23}. 
In order to terminate the recursion of the algorithm, $\Delta$ must satisfy the condition $\Delta \leq \epsilon'_1$. 
When the number of iterations is $n$, we can transpose the above condition to $n \geq \log_{2}{(\Delta_0 / \epsilon'_1)}$, where $\Delta_0$ is the initial value of $\Delta$. 
Therefore, the required number of iterations decreases by considerable amount, about 10 in our parameter configuration.
This difference accumulates for every candidate and generation of the algorithm, resulting in a significant time saving.

Even if we exclude the influence of the evaluation part with the simplified time slot allocation algorithm, our proposed scheme consists of sorting, mutation rate calculation, and operations that follow $O(1)$ computation time. 
On the other hand, the VAE scheme adopts fully connected layers composed of a number of weights proportional to $N_d^{2}$ and updates weights through back-propagation, which requires considerable complex calculations. 
Note that although other conventional schemes require less computation time than our proposed scheme due to their linear and straightforward computations, resulted rate budgets of conventional schemes are much lower than rate budgets of our proposed scheme. 
Thus, we can validate that our proposed scheme computes relaying topology with a sufficient rate budget in a reasonable computation time.

%%%%%%%%%%%%%%%%%%%%%%%%%%%%%%%%%%%%%%%%%%%%%%%%%%%%%%%%%%%%%%%%%%%%%%%%%%%%%
\section{Conclusions}
In this study, we formulated a system model which includes randomly distributed IoT devices and power beacons considering a TDMA system and energy harvesting. 
We proved the validity of iterative balancing time slot allocation algorithm using the KKT condition. 
We proposed a GMGA which modulates the mutation rate inversely proportional to the link cost for reasonable connection to overcome the limitations of a normal genetic algorithm; an excess of low-performance candidates, such as direct sink connections of low-powered nodes or connections to the distant nodes. 
We also proposed a mobility-aware iterative relaying topology algorithm which starts the calculation from the final result of the previous frame. 
To show the validity of our proposed algorithms, we conducted a stationary simulation environment and a mobility simulation environment. 

However, our research is based on the simulations performed on a system model written in MATLAB. 
Therefore, an implementation considering practical telecommunication conditions such as digital modulation techniques and coding rate could be helpful for the future works to gain further insight on relaying scheme and increase the reliability of the results.
In addition, future works might investigate the application of another deep learning techniques in order to enhance the performance of relaying network.

Nevertheless, the methodology proposed in our study is significant since it provides superiority in quantitative metrics compared to currently existing VAE based method or conventional relaying methods.
We also provided simulation results of our proposed scheme compared with conventional schemes and an exhaustive search scheme, which searches the optimal solution based on brute-force method. 
The simulation results show that the rate budget is increased by 11.75\% compared to VAE scheme when averaging the rate of increase for each data point. Also, the computation time is decreased by 87.70\% compared to VAE scheme when comparing the total computation time between two schemes.
In conclusion, we observed and confirmed that our proposed scheme outperforms other schemes in terms of both performance and computation time, giving us a glimpse of the feasibility of real-world adoption.

%%%%%%%%%%%%%%%%%%%%%%%%%%%%%%%%%%%%%%%%%%%%%%%%%%%%%%%%%%%%%%%%%%%%%%%%%%%%%

\bibliographystyle{IEEEtran}

%%%%%%%%%%%%%%%%%%%%%%%%%%%%%%%%%%%%%%%%%%%%%%%%%%%%%%%%%%%%%%%%%%%%%%%%%%%%%
\begin{IEEEbiography}[{\includegraphics[width=1in,height=1.25in,clip,keepaspectratio]{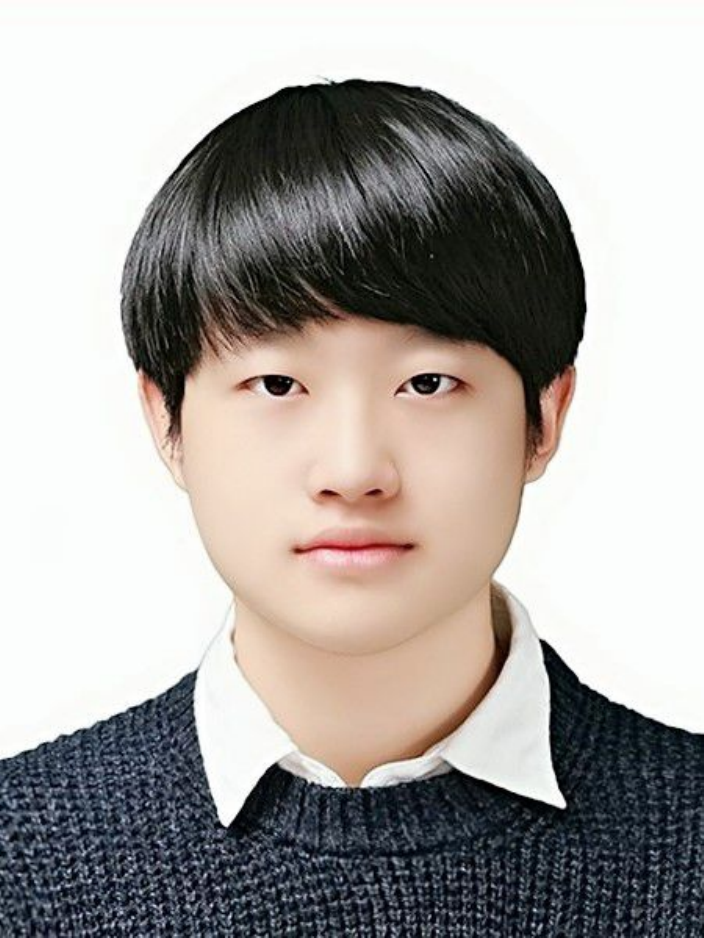}}]{Gyupil Kam} received a B.S degree in Electrical Engineering from Korea Advanced Institute of Science and Technology (KAIST), Daejeon, South Korea, in 2023. Since June 2023, he has been a research officer at the Agency for Defense Development (ADD), South Korea. His research interests include telecommunication, genetic algorithm, internet of things (IoT), wireless networks, and machine learning. 
\end{IEEEbiography}
\begin{IEEEbiography}[{\includegraphics[width=1in,height=1.25in,clip,keepaspectratio]{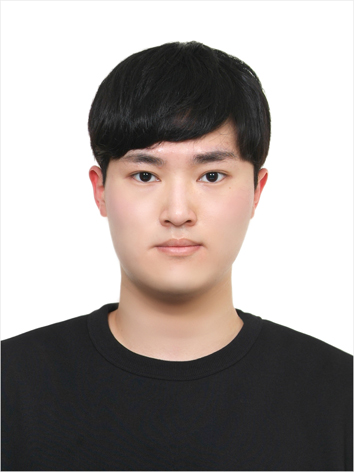}}]{Kiseop Chung} received a B.S degree in Electrical and Computer Engineering from Seoul National University, Seoul, South Korea, in 2022. Since June 2022, he has been a research officer at the Agency for Defense Development (ADD), South Korea. His research interests include internet of things (IoT), wireless networks, unsupervised machine learning, embedded systems, and hardware architecture.
\end{IEEEbiography}

\EOD

\end{document}